\documentclass[11pt,letterpaper]{article}
\pdfoutput=1
\usepackage{style}
\usepackage{shortcuts}

\crefname{claim}{Claim}{Claims}
\Crefname{claim}{Claim}{Claims}
\crefname{problem}{Problem}{Problems}
\Crefname{problem}{Problem}{Problems}

\title{Deterministic Sparse Pattern Matching\\via the Baur-Strassen Theorem}
\author{Nick Fischer\footnote{This work is part of the project CONJEXITY that has received funding from the European Research Council (ERC) under the European Union's Horizon Europe research and innovation programme (grant agreement No.~101078482).}\\[.2ex] (Weizmann Institute of Science)}
\date{}

\begin{document}
\maketitle

\begin{abstract}
\noindent
How fast can you test whether a constellation of stars appears in the night sky? This question can be modeled as the computational problem of testing whether a set of points $P$ can be moved into (or close to) another set $Q$ under some prescribed group of transformations. Problems of this kind are subject to intensive study in computational geometry and enjoy countless theoretical and practical applications.

Consider, as a simple representative, the following problem: Given two sets of at most~$n$ integers~$P, Q \subseteq [N]$, determine whether there is some shift $s$ such that $P$ shifted by $s$ is a subset of $Q$, i.e.,~\makebox{$P + s = \set{p + s : p \in P} \subseteq Q$}. This problem, to which we refer as the \emph{Constellation} problem, can be solved in near-linear time $\Order(n \log n)$ by a Monte Carlo randomized algorithm~[Cardoze, Schulman; FOCS~'98] and time $\Order(n \log^2 N)$ by a Las Vegas randomized algorithm~[Cole, Hariharan;~STOC~'02]. Moreover, there is a deterministic algorithm running in time~\smash{$n \cdot 2^{\Order(\sqrt{\log n \log\log N})}$}~[Chan, Lewenstein;~STOC~'15]. An interesting question left open by these previous works is whether Constellation is in \emph{deterministic near-linear} time (i.e., with only polylogarithmic overhead).

We answer this question positively by giving an $\Order(n \polylog(N))$-time deterministic algorithm for the Constellation problem. Our algorithm extends to various more complex Point Pattern Matching problems in higher dimensions, under translations and rigid motions, and possibly with mismatches, and also to a near-linear-time derandomization of the Sparse Wildcard Matching problem on strings.

We find it particularly interesting \emph{how} we obtain our deterministic algorithm. All previous algorithms are based on the same baseline idea, using additive hashing and the Fast Fourier Transform. In contrast, our algorithms are based on new ideas, involving a surprising blend of combinatorial and \emph{algebraic} techniques. At the heart lies an innovative application of the Baur-Strassen theorem from algebraic complexity theory.
\end{abstract}

\setcounter{page}{0}
\thispagestyle{empty}
\newpage

\section{Introduction}
Consider a constellation of stars; how fast can you test whether this constellation appears in the night~sky? This question, modeled in various flavors as the computational problem of testing whether a set of points~$P$ can be moved into (or close to) another point set $Q$ under some prescribed group of transformations, defines a rich family of pattern matching tasks from computational geometry and string algorithms. These problems have been extensively studied~\cite{AltMWW88,HuttenlocherKK92,Rucklidge93,GoodrichMO94,RezendeL95,Boxer96,IraniR96,ChewGHKKK97,CardozeS98,AkutsuTT98,GavrilovIMV99,IndykMV99,ColeH02,IndykV03,AmirKP07,ChoM08,AigerK09,Ukkonen10,ChanL15} and their applications within theoretical and practical computer science are numerous: In the context of image processing and computer vision, point pattern matching is employed for \emph{image registration}~\cite{MountNM98} (the task of determining for two images of the same scene which transformation most nearly maps one image into the other), and \emph{model-based object recognition}~\cite{Rucklidge96}. In computational chemistry, the applications include \emph{pharmacophore identification}~\cite{FinnKLMSVY97} and \emph{protein structure alignment}~\cite{Akutsu96}.

For the first few pages of this paper let us focus on a single basic representative of this colorful class of problems, to which we will simply refer as the \emph{Constellation} problem: Determine for two sets~\makebox{$P, Q \subseteq [N]$} of at most $n$ integers, whether there is a shift $s$ with~\makebox{$P + s = \set{p + s : p \in P} \subseteq Q$}. The perks of this problem are twofold: First, it is appealingly fundamental and simple to state. Second, the Constellation problem is in fact not just a toy problem, but captures the core hardness of many pattern matching problems on points and strings, as we will describe in \cref{sec:introduction:sec:results:sec:point,sec:introduction:sec:results:sec:string}. In particular, many (but not all) problems that we will discuss later even admit black-box reductions to the Constellation problem.

Despite its apparent simplicity, the Constellation problem is far from trivial and underwent an interesting history. The naive solution takes time $\Order(n^2)$.\footnote{Note that there are at most $|Q|$ candidate shifts $s$ (as an arbitrary point in $P$ can be mapped to any point in $Q$). Verifying the candidate shifts takes time $\Order(|P|)$ each. The total time is $\Order(|P| \cdot |Q|) = \Order(n^2)$.} In terms of randomized algorithms, it is well-established by two seminal papers from over 20 years ago that the Constellation problem can be solved in near-linear time. Specifically, Cardoze and Schulman~\cite{CardozeS98} designed a Monte Carlo randomized algorithm for the Constellation problem and many of its generalizations, running in near-linear time $\Order(n \log n)$. Their key idea was to use additive hashing to reduce the universe size from $N$ to $\Order(n)$, and to make use of the Fast Fourier Transform on these smaller universes. Shortly after, Cole and Hariharan~\cite{ColeH02} devised a verification approach leading to a Las Vegas randomized algorithm with near-linear running time $\Order(n \log^2 N)$. In the same paper, they managed to give near-linear \emph{deterministic} algorithms for many related pattern matching problems such as Wildcard Matching; however, finding a near-linear deterministic algorithm for Constellation remained open.

The first nontrivial deterministic algorithm, improving over the $\Order(n^2)$-time baseline solution, was suggested by Amir, Kapah and Porat~\cite{AmirKP07} who demonstrated how to preprocess the set $Q$ in time \smash{$\widetilde\Order(n^2)$} in order to support queries for a given set $P$ in near-linear time. Chan and Lewenstein~\cite{ChanL15} improved the preprocessing time to $\Order(n^{1+\epsilon})$, for any constant $\epsilon > 0$, using a very general technique of constructing pseudo-additive pseudo-perfect hash families. In particular, their result implies a deterministic algorithm for the Constellation problem in time $\Order(n^{1+\epsilon})$, for any constant~\makebox{$\epsilon > 0$}. A closer inspection reveals that they even obtain a deterministic algorithm in almost-linear time~\smash{$n \cdot 2^{\Order(\sqrt{\log n \log\log N})}$}.

Their result settles that the deterministic complexity is linear up to subpolynomial factors. However, the super-polylogarithmic overhead is quite unsatisfactory---especially since for many related pattern matching problems~\cite{ColeHI99,ColeH02}, and also for the related sparse convolution problem~\cite{BringmannFN22}, deterministic near-linear time algorithms are known. In this paper, we therefore revisit the following question:
\medskip
\begin{center}
    \emph{Can the Constellation problem be solved in deterministic near-linear time $\Order(n \polylog(N))$?}
\end{center}
\medskip
We remark that in general efficient derandomizations are sought after in the area of geometric pattern matching~\cite{Indyk97,ColeHI99,ColeH02,CliffordC07,Gawrychowski11}. Moreover, finding a deterministic near-linear algorithm for the Constellation problem is the only derandomization question which is still open from Cole and Hariharan's impactful paper~\cite{ColeH02}. 

\subsection{Our Core Results}
Our core result is that we answer our driving question affirmatively:

\begin{theorem}[Deterministic Constellation] \label{thm:constellation}
Given two sets $P, Q \subseteq [N]$ of size at most $n$, we can list all shifts $s$ satisfying $P + s \subseteq Q$ in deterministic time $\Order(n \polylog(N))$.
\end{theorem}

Recall that this improves the subpolynomial overhead \smash{$2^{\Order(\sqrt{\log n \log\log N})}$} from the previously best deterministic algorithm to a polylogarithmic overhead. (The fastest known randomized algorithms are still faster in terms of the number of log-factors.) As an additional feature, our algorithm (as well as its predecessors) not only decides the existence of a feasible shift $s$, but also \emph{lists} all feasible shifts.

In contrast to the previous works~\cite{CardozeS98,ColeH02,AmirKP07,ChanL15} which basically follow the same baseline idea, our algorithm is based on completely new ideas, involving a surprising blend of combinatorial and \emph{algebraic} techniques. Moreover, up to technical complications in the underlying algebra machinery, in its core our algorithm is arguably simple. In \cref{sec:introduction:sec:technical} we will sketch our algorithm in some detail, highlighting in particular our surprising usage of a theorem from algebraic complexity theory: the Baur-Strassen theorem.

\paragraph{With \boldmath$k$ Mismatches}
Before, we discuss some interesting generalizations of \cref{thm:constellation}. Note that in the Constellation problem we are very strict in what we consider a feasible shift $s$---namely, $P + s$ must be an exact subset of $Q$. However, there are many applications for which we consider $s$ a feasible shift under more relaxed assumptions. For instance, suppose that some stars in the pattern constellation are incorrect\footnote{Perhaps they were mistaken for airplanes.}, and we would only like to test whether the constellation appears in the night sky \emph{up to $k$ mismatches.} Formally, the goal is to report shifts $s$ with $|(P + s) \setminus Q| \leq k$. This problem was also studied by Cardoze and Schulman~\cite{CardozeS98}. They obtained a Monte Carlo randomized algorithm in time\footnote{To avoid verbose notation, we implicitly understand that $k$ means $\max\set{1, k}$ in the $\Order$-notation.} \smash{$\widetilde\Order(n k)$}, subject to the technical assumption that $k$ is smaller than a constant fraction of $|P|$, $k \leq (1 - \Omega(1)) |P|$. (This assumption implies in particular that the number of solutions is at most $\Order(n)$.)

In stark contrast to the Constellation problem without mismatches, for the mismatch version there are no non-trivial Las Vegas algorithms let alone deterministic algorithms known, to the best of our knowledge. In particular, the best known deterministic algorithm is the $\Order(n^2 k)$-time naive solution.\footnote{This $\Order(n^2 k)$-time algorithm works as follows: Fix $k + 1$ arbitrary points in $P$, called the \emph{anchor} points. Under each feasible shift, at least one of the anchor points matches with $Q$. Therefore, we can enumerate the $n(k + 1)$ anchor-to-$Q$ alignments, and test each alignment in time $\Order(n)$.} Our contribution is that we give a deterministic algorithm that is faster by a linear factor, and matches the randomized time complexity of~\cite{CardozeS98} at the cost of a larger polylogarithmic term:

\begin{theorem}[Deterministic Constellation with Mismatches] \label{thm:constellation-mismatches}
Given sets $P, Q \subseteq [N]$ of size at~most~$n$ and $0 \leq k \leq (1 - \Omega(1))|P|$, we can list all shifts $s$ with $|(P + s) \setminus Q| \leq k$ in deterministic time $\Order(n k \polylog(N))$.
\end{theorem}

This theorem is most effective in the relevant setting where $k$ is comparably small. In particular, we can compute all shifts $s$ with at most $\polylog(N)$ mismatches in near-linear deterministic time $\Order(n \polylog(N))$. We remark that our algorithm works as well for a \emph{weighted} version, where the goal is to tolerate mismatches with total weight up to $k$. We defer the precise statement to the technical sections; see \cref{thm:constellation-weights}.

\paragraph{Realistic Point Pattern Matching}
Even taking mismatches into account, the Constellation problem treats a very restricted pattern matching setting. In reality, there are many other reasonable ways of further relaxing what we consider a feasible shift. For instance, an actual constellation of stars consists of \emph{2-dimensional} points, and it is reasonable to take \emph{rotations} into account. Moreover, in most applications the input data is not perfectly precise, so it is reasonable to consider \emph{approximate matches.}

Our algorithm naturally extends to many of these settings. We will precisely state our results in \cref{sec:introduction:sec:results:sec:point}, starting with a systematic definition of Point Pattern Matching. Interestingly, all of our results in the following sections can be proven via black-box reductions to \cref{thm:constellation} or \cref{thm:constellation-mismatches}. This approach---reducing realistic Point Pattern Matching problems to the Constellation problem---was also taken by Cardoze and Schulman~\cite{CardozeS98}, and we adopt many of their reductions without any changes. Others require (minor) modifications to avoid randomization.

\subsection{Consequences for Point Pattern Matching} \label{sec:introduction:sec:results:sec:point}
Recall that the general task in Point Pattern Matching is to test whether there is a transformation that carries a point set $P$ into or near some other point set~$Q$. This problem has been studied in several variations; here we consider the following four design choices.
\begin{itemize}
    \item\emph{Dimension $d$:} Throughout, we assume that $P, Q$ contain $d$-dimensional points. Most works study Point Pattern Matching for constant dimension $d = 1, 2, 3$.
    \item\emph{Transformation group:} Here we consider either \emph{translations} or \emph{rigid motions} (i.e., translations plus rotations). These are the most commonly studied transformations groups, but there are also works devoted to even richer transformations, involving e.g.\ scaling.
    \item\emph{Exact versus approximate:} In many applications we would like to classify transformations of~$P$ as matches even if they slightly differ from $Q$. We will focus on the following two settings:
    \begin{itemize}
        \item \emph{Exact:} Is there a transformation of $P$ that is a subset of $Q$? This is the most fundamental similarity measure and has been extensively studied before~\cite{AltMWW88,RezendeL95,Boxer96,IraniR96,AkutsuTT98,CardozeS98,IndykMV99,ColeH02}.
        \item \emph{Approximate under the Hausdorff distance:} For $\epsilon, \delta > 0$, distinguish whether there is a transformation of $P$ that has directed Hausdorff distance at most $\delta$ to $Q$, versus any transformation of $P$ has directed Hausdorff distance more than $(1 + \epsilon) \delta$. This problem was studied e.g.\ in \cite{CardozeS98,IndykMV99}.
    \end{itemize}
    There are many other reasonable notions of exact and approximate matches. For instance, the Point Pattern Matching with exact thresholds under the Hausdorff distance (i.e., with $\delta > 0$ and $\epsilon = 0$) was studied in \cite{HuttenlocherKK92,Rucklidge93,ChewGHKKK97}, leading to algorithms with a significant polynomial overhead. The main advantage of the approximate version is that it allows for a reduction to the exact case, where the overhead depends only on the secondary parameters $\delta$ and $\epsilon$. Another popular similarity measure is the \emph{bottleneck distance} (which requires that each point in $P$ is matched to a unique point in $Q$)~\cite{AltMWW88,IndykMV99}.

    \item \emph{Mismatches:} While the Hausdorff distance nicely captures that two sets are close up to small perturbations, it does not account for rare outliers. The typical model, as mentioned before, is to consider up to $k$ \emph{mismatches}; i.e., we only require that $P$ matches $Q$ after deleting the~$k$ worst points from $P$. This problem was studied in~\cite{AkutsuTT98,CardozeS98,IndykMV99}, often under the alias \emph{Largest Common Point Set}.
\end{itemize}
A final less important distinction is whether we consider integer or real points. Following~\cite{CardozeS98} we consider integer points for all exact problems, and bounded-precision reals for all approximate problems.
We defer the formal definitions of these problems to \cref{sec:corollaries}.

The entire landscape of algorithms for these Point Pattern Matching problems is complex and too large to be summarized here. Instead, we will closely follow the presentation of Cardoze and Schulman's paper~\cite{CardozeS98} and consider five specific variants. In short, our contribution is that we replicate \emph{all} results from~\cite{CardozeS98} with deterministic algorithms, at the cost of worsening the running times by $\poly(d \log N)$.

\paragraph{Without Mismatches}
We first consider the case without mismatches. Cardoze and Schulman established that Point Pattern Matching with translations is in near-linear randomized time (for both exact and approximate matches). For rigid motions their running time is \smash{$\widetilde\Order(n^d)$} (for approximate matches, the only reasonable variant for rotations) which is faster than the brute-force~$\Order(n^{d+1})$ time by a linear factor. As for the Constellation problem, there are Las Vegas algorithms with overhead $\polylog(N)$~\cite{ColeH02} and deterministic algorithms with overhead \smash{$2^{\Order(\sqrt{\log n \log\log N})}$}~\cite{AmirKP07,ChanL15}. Our contribution is again that we reduce the overhead from subpolynomial~\smash{$2^{\Order(\sqrt{\log n \log\log N})}$} to polylogarithmic:

\begin{corollary}[Exact Point Pattern Matching with Translations] \label{cor:ppm-translations-exact}
The exact Point Pattern Matching problem with translations is in deterministic time $\Order(n \poly(d \log N))$.
\end{corollary}

\begin{corollary}[Approximate Point Pattern Matching with Translations] \label{cor:ppm-translations-approx}
The approximate Point Pattern Matching problem with translations is in deterministic time $\Order(n \epsilon^{-\Order(d)} \polylog(N \delta^{-1}))$.
\end{corollary}

\begin{corollary}[Approximate Point Pattern Matching with Rigid Motions] \label{cor:ppm-rigid-approx}
The approximate Point Pattern Matching problem with rigid motions is in deterministic time $\Order(n^d \epsilon^{-\Order(d)} \polylog(N \delta^{-1}))$.
\end{corollary}

\paragraph{With \boldmath$k$ Mismatches}
Consider again the more general Point Pattern Matching problems with up to $k$ mismatches. For translations, the known Monte Carlo algorithms run in time $\widetilde\Order(nk)$~\cite{CardozeS98}. We remark that, as for the Constellation problem, for Point Pattern Matching with mismatches there are no known Las Vegas algorithms or derandomizations, and so the best-known deterministic approach is the naive $\Order(n^2 k)$-time algorithm. Our deterministic algorithms match the randomized bounds up to logarithmic factors, and therefore improve what is known by nearly a linear factor:

\begin{corollary}[Exact Point Pattern Matching with Mismatches] \label{cor:ppm-mismatches-exact}
The exact Point Pattern Matching problem with translations and up to $0 \leq k \leq (1 - \Omega(1)) |P|$ mismatches is in deterministic time $\Order(n k \poly(d \log N))$.
\end{corollary}

\begin{corollary}[Approximate Point Pattern Matching with Mismatches] \label{cor:ppm-mismatches-approx}
The approximate Point Pattern Matching problem with translations and up to $0 \leq k \leq (1 - \Omega(1)) |P|$ mismatches is in deterministic time $\Order(n k \epsilon^{-\Order(d)} \polylog(N \delta^{-1}))$.    
\end{corollary}

\subsection{Consequences for Sparse Wildcard Matching} \label{sec:introduction:sec:results:sec:string}
Another interesting consequence of our results concerns the \emph{Sparse Wildcard Matching} problem. In the \emph{Wildcard Matching} problem, the task is to check whether a pattern text $P \in (\Sigma \cup \set{*})^M$ matches a substring of some text $T \in \Sigma^N$, where the wildcard ``$*$'' is understood to match any character from the alphabet $\Sigma$. This problem has an interesting history, starting with a classical algorithm due to Fischer and Paterson~\cite{FischerP74} in time~\makebox{$\Order(N \log M \log |\Sigma|)$}, and leading to improvements in deterministic time~\makebox{$\Order(N \log M)$}~\cite{Indyk98a,Kalai02,ColeH02}. In the \emph{Sparse} Wildcard Matching problem the task is the same; however, we assume that the pattern and text are both sparse in the following sense: The text contains at most $n$ non-zero characters (for some designated character ``zero''), and the pattern contains no zero characters and at most $n$ non-wildcard characters.

This problem was introduced by Cole and Hariharan~\cite{ColeH02} as a string pattern matching problem that is \emph{equivalent} to the Constellation problem. It thus shares the same history, with an $\Order(n \log^2 N)$ Las Vegas algorithm~\cite{ColeH02}, and an \smash{$n \cdot 2^{\Order(\sqrt{\log n \log\log N})}$}-time deterministic algorithm~\cite{AmirKP07,ChanL15}. Besides, this problem was studied in parallel models of computation~\cite{HajiaghayiSSS21}. As a direct corollary of \cref{thm:constellation}, we obtain that the Sparse Wildcard Matching problem is in deterministic near-linear time.

\begin{corollary}[Sparse Wildcard Matching] \label{cor:sparse-wildcard-matching}
The Sparse Wildcard Matching problem is in deterministic time $\Order(n \polylog(N))$.
\end{corollary}

\subsection{Technical Highlights} \label{sec:introduction:sec:technical}
In this section we briefly sketch the technical highlights of our paper. Our approach can be neatly split into two independent steps.

\paragraph{Step 1: From Pattern Matching to Convolutions}
Almost 50 years ago, Fischer and Paterson~\cite{FischerP74} proposed a very general framework for solving (string) pattern matching problems. Their idea was to reduce pattern matching to computing a small number of \emph{convolutions} (the convolution $x \conv y$ of two vectors $x$ and $y$ is defined via $(x \conv y)[k] = \sum_{i + j = k} x[i] \cdot y[j]$), that can be implemented in near-linear time by the Fast Fourier Transform. The approach applies in various contexts involving also matching with wildcards, and has been extended to many other problems since~\cite{Abrahamson87,AmirF95,AmirL91,GalilG88,Kosaraju89,Pinter85,MuthukrishnanR95,MuthukrishnanP94,Muthukrishnan95,CardozeS98,ColeH02}.

In light of this, it seems plausible that \emph{sparse} pattern matching problems (like Constellation) admit reductions to \emph{sparse} convolution problems. Muthukrishnan~\cite{Muthukrishnan95} suggested two sparse convolution variants both of which seem plausible targets for a reduction:
\begin{itemize}
\item In the \emph{Sparse Convolution} problem the task is to compute the convolution $x \conv y$ of two (nonnegative) integer vectors $x, y \in \Int^N$ in time proportional to the number of non-zeros in~$x$,~$y$ and $x \conv y$.
\item In the \emph{Partial Convolution} problem we are additionally given a set~\makebox{$C \subseteq [N]$} and the goal is to compute the convolution restricted to the coordinates in $C$.
\end{itemize}
By now, both problems have received considerable attention in the literature with quite different outcomes: On the one hand, the Sparse Convolution problem underwent a series of algorithmic improvements~\cite{ColeH02,AmirKP07,ArnoldR15,ChanL15,Nakos20,GiorgiGC20,BringmannFN21,BringmannFN22} and is known to be in near-linear output-sensitive time~\cite{ColeH02,Nakos20}, even by a deterministic algorithm~\cite{BringmannFN22} (for nonnegative vectors). On~the~other hand, via a simple reduction, the Partial Convolution problem requires time~\makebox{$n^{2-\order(1)}$} (where~\makebox{$n = \norm{x}_0 + \norm{y}_0 + |C|$} is the input size) unless the 3SUM conjecture from fine-grained complexity theory fails; see also~\cite{GoldsteinKLP16}.

In light of this, a reasonable approach for the Constellation problem is to look for a reduction to the Sparse Convolution problem. However, despite much focus and the apparent connection that both problems share (for instance in terms of algorithmic tools like additive hashing and FFT), previous work fell short of giving the desired reduction. 

Our first contribution is that we establish the missing reduction, but with an unexpected twist: We prove that the Constellation problem reduces to a logarithmic number of instances of the \emph{Partial} Convolution problem on \emph{structured} instances. Recall that, since it is unlikely that Partial Convolutions can be computed faster than quadratic time on worst-case instances, any efficient algorithmic reduction to the Partial Convolution problem is forced to produce structured instances. The structure is as follows: Let $A, B \subseteq [N]$ denote the support of $x, y$, respectively (i.e., the set of nonzero coordinates). Then the reduction guarantees that $|C - B|$ is small, where~\makebox{$C - B = \set{c - b : c \in C,\, b \in B}$} is the \emph{difference set} of $B$ and $C$. Formally:

\begin{restatable}[Constellation to Partial Convolution]{lemma}{lemconstellationtopartialconv} \label{lem:constellation-to-partial-conv}
If the Partial Convolution problem is in deterministic time $\Order((|A| + |C - B|) \cdot \log^c N)$, then the Constellation problem is in deterministic time~$\Order(n \log^{c+1} N)$.
\end{restatable}

While this running time for Partial Convolution is not ruled out by the 3SUM lower bound (as in worst-case instances we have $|B - C| = \Omega(n^2)$), it seems nevertheless quite unclear how to achieve this running time. Using the deterministic algorithm for Sparse Convolution we can achieve time \smash{$\widetilde\Order(|A + B| + |C|)$}, which does not seem helpful here.\footnote{It is easy to construct sets $A, B, C$ for which $|A| + |C - B| = \Order(n)$ and $|A + B| + |C| = \Omega(n^2)$: Let $B = C = \set{1, 2, \dots, n}$ and let $A = \set{n, 2n, \dots, n^2}$.}

Our remaining goal, which we will soon complete in step 2 of the algorithm, is to design an algorithm for Partial Convolution with near-linear running time in $|A| + |B - C|$. Before, let us take a brief intermezzo to introduce our key tool:

\paragraph{The Baur-Strassen Theorem}
The Baur-Strassen theorem is a powerful tool from algebraic complexity, stating intuitively that we can compute the partial derivatives of a function with the same complexity as the function itself:

\begin{restatable}[Baur-Strassen~\cite{BaurS83,Morgenstern85}]{theorem}{thmbaurstrassen} \label{thm:baur-strassen}
For any arithmetic circuit $C$ computing $f(x_1, \dots, x_n)$, there is a circuit $C'$ that simultaneously computes the partial derivatives $\frac{\partial f}{\partial x_i}(x_1, \dots, x_n)$ for all $1 \leq i \leq n$. The circuit $C'$ has size $|C'| \leq \Order(|C|)$ and can be constructed in time $\Order(|C|)$.
\end{restatable}

Baur and Strassen established this result in 1983~\cite{BaurS83}, and a simplified constructive proof followed shortly after~\cite{Morgenstern85}. Their original intention was to prove arithmetic circuit \emph{lower bounds} (in particular, the Baur-Strassen theorem implies that matrix multiplication and determinant computation have the same asymptotic complexity). But also in algorithm design, the Baur-Strassen theorem had an exciting appearance: Cygan, Gabow and Sankowski designed simple and fast algorithms for several graph problems including finding shortest cycles, computing the radius and diameter, and finding minimum-weight perfect matchings~\cite{CyganGS15}. Their algorithms are based on the Baur-Strassen theorem combined with fast matrix multiplication. Our work constitutes only the second application in algorithm design that we are aware of, and the first application in combination with an FFT-like circuit.

Intuitively, the Baur-Strassen theorem allows to \emph{invert the flow} of a circuit, and to propagate information from the outputs to the inputs---After all, this is exactly what a partial derivative does: It measures how the output is affected by an input. Under the name \emph{back-propagation} this usage of the Baur-Strassen theorem is omni-present in machine learning~\cite{Werbos74,Rumelhart86,Werbos94}. More relevant to us is the consequence that a bilinear circuit computing a function $z_k = \sum_{i, j} a_{i, j, k} \cdot x_i \cdot y_j$ can be transformed into a circuit computing the ``rotated'' function $x_i = \sum_{j, k} a_{i, j, k} \cdot y_j \cdot z_k$; details follow in the paragraph. In contrast to the typical use cases in algebraic complexity, we here critically exploit that the new circuit can also be efficiently constructed.

\paragraph{Step 2: Partial Convolutions via the Baur-Strassen Theorem}
Let us come back to the algorithm. The Baur-Strassen theorem is only useful in combination with an arithmetic circuit; in our case that circuit should naturally compute the convolution of two sparse vectors. Such circuits exist, and they have already been used by Bringmann, Fischer and Nakos~\cite{BringmannFN22} to obtain a deterministic near-linear algorithm for the Sparse Convolution problem (on nonnegative vectors). Specifically, the sparse convolution technology yields a circuit of size~\smash{$\widetilde\Order(|C - B|)$} with inputs~$y_b$ (for~$b \in B$) and $z_c$ (for $c \in C$) and outputs $w_a$ (for $a \in C - B$) that computes
\begin{equation*}
    w_a = \sum_{\substack{b \in B\\c \in C\\a = c - b}} y_b \cdot z_c.
\end{equation*}
We will now modify this circuit in two steps. First, we add new inputs $x_a$ (for $a \in A$) and with small overhead let the circuit compute
\begin{equation*}
    w := \sum_{a \in A \cap (C - B)} x_a \cdot w_a = \sum_{\substack{a \in A\\b \in B\\c \in C\\a + b = c}} x_a \cdot y_b \cdot z_c.
\end{equation*}
Then observe that the partial derivatives
\begin{equation*}
    \frac{\partial w}{\partial z_c} = \sum_{\substack{a \in A\\b \in B\\a + b = c}} x_a \cdot y_b
\end{equation*}
are exactly what we want to compute. We can thus apply the Baur-Strassen theorem to let our circuit compute the derivatives of $w$ with respect to all inputs $z_c$. The size has not increased asymptotically and is still \smash{$\widetilde\Order(|C - B|)$}.

Following this insight, our deterministic algorithm explicitly constructs this arithmetic circuit (applying the Baur-Strassen theorem in the process), and evaluates it on the given vectors $x, y$. Unfortunately, while this overview seems very simple, the implementation of this idea suffers from many technical complications due to insufficiencies in the underlying computer algebra machinery. One of the main problems is that we need to perform all computations over a finite field, but we cannot assume deterministic access to a prime number of size $N$. Luckily, we can deal with these problems as in~\cite{BringmannFN22}; see \cref{sec:conv} for a more detailed discussion.

\paragraph{3SUM on Structured Instances}
A notable consequence of this idea is for the 3SUM problem on structured instances. The 3SUM problem is to decide for three integer sets $A, B, C$ of size at most $n$, whether there is a triple $(a, b, c) \in A \times B \times C$ with $a + b + c = 0$. It is a central conjecture in fine-grained complexity theory that this problem requires quadratic time $n^{2-\order(1)}$~\cite{GajentaanO95}, and many conditional hardness results are known based on this conjecture. On the positive side, it is known that for structured sets, 3SUM can be solved faster. Specifically, there is a (deterministic) algorithm in time $\widetilde\Order(n + \min\set{|A + B|, |A + C|, |B + C|})$ that simply computes the smallest of the three sumsets via a sparse convolution algorithm. 3SUM algorithms of this type have recently found applications in fine-grained lower bounds~\cite{AbboudBF23,JinX23}. With our novel approach, we can solve the following ``count-all-numbers'' version of 3SUM in the same running time; prior to our work this statement was not even known in terms of randomized algorithms, as far as we are aware:

\begin{theorem}[\#AllNumbers3SUM] \label{thm:all-numbers-3sum}
Given three sets $A, B, C \subseteq \set{-n^{\Order(1)}, \dots, n^{\Order(1)}}$ of size at most $n$, we call a triple $(a, b, c) \in A \times B \times C$ that satisfies $a + b + c = 0$ a \emph{3-sum}. There is an algorithm that computes for each $z \in A \cup B \cup C$ in how many 3-sums it participates, and runs in deterministic time~\smash{$\widetilde\Order(n + \min\set{|A + B|, |A + C|, |B + C|})$}.
\end{theorem}

\subsection{Open Questions}
Our work inspires some interesting open questions.
\begin{enumerate}
    \setlength\parskip{0ex}
    \item\emph{Can the number of log-factors be improved?} It can be checked that our Constellation algorithm runs in time $\Order(n \log^6 N \polyloglog N)$. Shaving one log-factor is easy by merging two ``scaling steps'' in the algorithm,\footnote{More concretely, one can directly combine the upcoming \cref{lem:partial-conv-superset} with the reduction in \cref{lem:constellation-to-partial-conv}. The superset $T$ of~$C - B$ can be maintained throughout the scaling without additional cost. For the sake of simplicity, we decided to stick to this more modular presentation of our result.} but any further improvement would require improving the algebra machinery (see also \cref{sec:conv}). We remark that if a prime $p > N$ and an element with multiplicative order at least $N$ in $\Field_p$ were provided in advance (both can be efficiently precomputed by a Las Vegas randomized algorithm), our algorithm would in fact run in time $\Order(n \log^2 n \log N)$. It is an interesting question whether this running time can be achieved without any assumption, or whether the number of log-factors can even be further reduced?
    \item\emph{Is there an algorithm for the Constellation problem that avoids FFT?}\\
    All Constellation algorithms faster than the quadratic-time baseline solution make direct or indirect use of the FFT. Is there a completely combinatorial algorithm for the constellation problem that, instead of FFT, leverages the structure of the sets $P$ and $Q$?
    \item\emph{Are there more algorithmic applications of the Baur-Strassen theorem?}
\end{enumerate}

\subsection{Outline}
We start with some preliminary definitions in \cref{sec:preliminaries}. In \cref{sec:constellation} we provide the reduction from the Constellation problem to the (structured) Partial Convolution problem. In \cref{sec:conv} we then give our algebraic algorithm for the Partial Convolution problem. Finally, in \cref{sec:corollaries} we state the reductions from Point Pattern Matching and Sparse Wildcard Matching to the Constellation problem.
\section{Preliminaries} \label{sec:preliminaries}
Let $\Int$, $\Int_{> 0}$, $\Real$ denote the integers, positive integers, and reals, respectively, and let $\Field_q$ denote the finite field with~$q$ elements. We write $[N] = \set{0, \dots, N - 1}$ and index all objects (such as vectors) starting with $0$, unless stated otherwise. Finally, we write $\poly(N) = N^{\Order(1)}$, $\polylog(N) = (\log N)^{\Order(1)}$, $\polyloglog(N) = (\log \log N)^{\Order(1)}$, and \smash{$\widetilde\Order(T) = T (\log T)^{\Order(1)}$}. Let us also formally recap the definitions of the Constellation problem, with and without mismatches:

\begin{problem}[Constellation]
Given sets $P, Q \subseteq [N]$, report all shifts $s \in \Int$ that satisfy~\makebox{$P + s \subseteq Q$}.
\end{problem}

\begin{problem}[Constellation with Mismatches] \label{prob:constellation-mismatches}
Given sets $P, Q \subseteq [N]$ and a threshold $0 \leq k < |P|$, report all shifts $s \in \Int$ satisfying $|(P + s) \setminus Q| \leq k$.
\end{problem}

\begin{problem}[Constellation with Weighted Mismatches] \label{prob:constellation-weighted}
Given sets $P, Q \subseteq [N]$, positive weights \makebox{$w : P \to \Int_{> 0}$} and a threshold $0 \leq k < w(P) = \sum_{p \in P} w(p)$, report all shifts $s \in \Int$ satisfying
\begin{equation*}
    \sum_{\substack{p \in P\\p + s \not\in Q}} w(p) \leq k.
\end{equation*}
\end{problem}

Note that \cref{prob:constellation-weighted} indeed generalizes \cref{prob:constellation-mismatches} as by picking unit-weights we recover exactly the condition $|(P + s) \setminus Q| \leq k$.

\paragraph{Sumsets and Convolutions}
Let $A, B \subseteq \Int$. We write $A + B = \set{a + b : a \in A,\, b \in B}$ to denote the \emph{sumset} of $A$ and $B$. For $s \in \Int$ we also write $A + s = \set{a + s : a \in A}$. For two vectors $x, y \in \Int^N$, the \emph{convolution $x \conv y \in \Int^{2N-1}$} is defined coordinate-wise by
\begin{equation*}
    (x \conv y)[k] = \sum_{\substack{i, j \in [N]\\i + j = k}} x[i] \cdot y[j].
\end{equation*}
We also define the \emph{cyclic convolution $x \conv_N y \in \Int^N$} analogously with wrap-around:
\begin{equation*}
    (x \conv_N y)[k] = \sum_{\substack{i, j \in [N]\\i + j \equiv k \mod N}} x[i] \cdot y[j].
\end{equation*}
Finally, for a vector $x \in \Int^n$, let $\supp(x) = \set{i \in [N] : x[i] \neq 0}$ denote its \emph{support.} Recall that the (cyclic) \emph{Partial Convolution} problem is to compute the (cyclic) convolution of two given vectors~$x, y$, restricted to some specific points. Formally: 

\begin{problem}[Partial Convolution]
Given $A, B, C \subseteq [N]$ and $x, y \in \Int^N$ with $\supp(x) \subseteq A$, $\supp(y) \subseteq B$, compute $(x \conv y)[c]$ for all $c \in C$.
\end{problem}

\begin{problem}[Cyclic Partial Convolution]
Given $A, B, C \subseteq [N]$ and $x, y \in \Int^N$ with $\supp(x) \subseteq A$, $\supp(y) \subseteq B$, compute $(x \conv_N y)[c]$ for all $c \in C$.
\end{problem}

For both problems, unless stated otherwise, we assume that the entries of $x$ and $y$ are bounded by $\poly(N)$. It is a simple observation that the non-cyclic and cyclic Partial Convolution problems are asymptotically equivalent for any reasonable parameterization, and we will therefore use both versions interchangeably.\footnote{For the reductions in both directions, it suffices to increase $N$ to $2N$.}

\paragraph{Machine Model}
We work in the word RAM model with word size $\Theta(\log N)$. In this model we can perform basic logical and arithmetic operations on words in constant time. In particular, in~$\Order(1)$ words we can simulate a real number with precision bounded by $\Order(\log N)$ bits.
\section{From Constellation to Partial Convolution} \label{sec:constellation}
In this section we give our reduction from the Constellation problem to Partial Convolution. We start with the reduction in \cref{lem:constellation-to-partial-conv}, and then strengthen the reduction in \cref{sec:constellation:sec:mismatches} to support mismatches.

\lemconstellationtopartialconv*
\begin{proof}
Throughout this proof, let us consider the Constellation problem for cyclic groups $\Int / N \Int$. That is, we are given sets $P, Q \subseteq \Int / N \Int$ and the task is to list all shifts $s \in \Int / N \Int$ satisfying $P + s \subseteq Q \mod N$ (i.e., that for all $p \in P$ there exists some $q \in Q$ with $p + s \equiv q \mod N$). Let us further assume that $N$ is a power of two. An algorithm for this problem directly leads to an algorithm for the standard Constellation problem with the same asymptotic running time, by choosing~$N$ to be the smallest power of two larger than $2 \cdot \max(P \cup Q)$ (as for this choice of $N$ there is no wrap-around).

The key idea behind the reduction is to apply a \emph{scaling trick.} If $N$ is smaller than some constant, we solve the instance by brute-force. Otherwise we will \emph{recursively} solve a Constellation instance over a smaller cyclic group of size $N' = \frac N2$, to obtain a good approximation to the original instance. More specifically, construct the sets
\begin{align*}
    P' &= \set{p \bmod N' : p \in P} \subseteq \Int / N' \Int, \\
    Q' &= \set{q \bmod N' : q \in Q} \subseteq \Int / N' \Int.
\end{align*}
We view $(P', Q')$ as a Constellation instance and solve it recursively. In this way we compute the maximal set~$S'$ satisfying~\makebox{$P' + S' \subseteq Q'$}. Note that whenever $s$ is a feasible shift in the original instance $(P, Q)$, then~$s \bmod N'$ is a feasible shift in the new instance $(P', Q')$ (but not necessarily vice versa). Therefore, the set $S' + \set{0, N'}$ surely contains all feasible shifts, plus possibly some false positives, and the remaining task is to filter all non-feasible shifts. To this end, we construct a Partial Convolution instance with
\begin{align*}
    A &= Q \subseteq \Int / N \Int, & x &= \One_Q \quad\text{(the indicator vector of $Q$)}, \\
    B &= -P \subseteq \Int / N \Int, & y &= \One_{-P} \quad\text{(the indicator vector of $-P$)}, \\
    C &= S' + \set{0, N'} \subseteq \Int / N \Int.
\end{align*}
Using the oracle, we can efficiently compute $(x \conv_N y)[s]$ for all $s \in C$. (Here we use the previous observation that computing cyclic and non-cyclic convolutions is equivalent.) Finally, we return
\begin{equation*}
    S = \set{s \in C : (x \conv_N y)[s] = |P|}.
\end{equation*}
This completes the description of the algorithm; in the next two steps we will prove that it is correct and efficient.

\begin{claim}[Correctness] \label{lem:constellation-to-partial-conv:clm:correctness}
$S = \set{s \in \Int / N \Int : P + s \subseteq Q \mod N}$.
\end{claim}
\begin{proof}
As mentioned before, for any shift $s$ with $P + s \subseteq Q \mod N$ it is clear that $P' + (s \bmod N') \subseteq Q' \mod{N'}$. Therefore, $C$ is a superset of the feasible shifts. Consider any shift~\makebox{$s \in C$}. By a straightforward calculation, we have that
\begin{equation*}
    (x \conv_N y)[s] = \sum_{\substack{a, b \in \Int / N \Int\\a + b \equiv s \mod N}} x[a] \cdot y[b] = \sum_{\substack{q \in Q\\p \in P\\q-p\equiv s \mod N}} 1 = \sum_{\substack{p \in P\\q \in Q\\p + s \equiv q \mod N}} 1.
\end{equation*}
That is, $(x \conv_N y)[s]$ equals the number of solutions to the equation $p + s \equiv q \mod N$, for $p \in P$ and~\makebox{$q \in Q$}. If~$s$ is feasible, then clearly there must be a solution for every $p \in P$, and the number of solutions is exactly~$|P|$. On the other hand, if $s$ is not feasible, then there exists some $p \in P$ with~\makebox{$p + s \not\in Q$}, and therefore the number of solutions is strictly less than $|P|$.
\end{proof}

\begin{claim}[Running Time] \label{lem:constellation-to-partial-conv:clm:time}
The algorithm runs in time $\Order(n \log^{c+1} N)$.
\end{claim}
\begin{proof}
Recall that solving the Partial Convolution instance takes time $\Order((|A| + |C - B|) \cdot \log^c N)$. While it looks like $|C - B|$ could have quadratic size, due to the way that we chose $C$ we can argue that $|C - B| = \Order(n)$. Indeed, recalling that $C = S' + \set{0, N'}$, $B = -P = -(P' + \set{0, N'})$ and $P' + S' \subseteq Q'$, we have that
\begin{equation*}
    C - B = S' + P' + \set{0, N', 2N'} \subseteq Q' + \set{0, N', 2N'}.
\end{equation*}
Since $|Q'| = |Q| \leq n$, we have that $|C - B| \leq 3n$ as claimed. This proves that a single call to the Partial Convolution oracle takes time $\Order(n \log^c N)$. Moreover, since the algorithm reaches recursion depth at most~$\log N$, the total running time is $\Order(n \log^{c+1} N)$ as claimed.
\end{proof}

\noindent
In combination, \cref{lem:constellation-to-partial-conv:clm:correctness,lem:constellation-to-partial-conv:clm:time} complete the proof of \cref{lem:constellation-to-partial-conv}.
\end{proof}

\subsection{Constellation with Mismatches} \label{sec:constellation:sec:mismatches}
Next, in \cref{lem:constellation-mismatches-to-partial-conv}, we strengthen the previous reduction to deal with mismatches. The following simple lemma about the number of solutions in the presence of mismatches will come in handy. Throughout, we deal with the more general case of \emph{weighted} mismatches (see \cref{prob:constellation-weighted}); the unweighted case is the restriction to unit weights ($w(p) = 1$).

\begin{lemma} \label{lem:constellation-mismatches-number}
Let $P, Q \subseteq [N]$ be sets, let $w : P \to \Int_{\geq 0}$ and let $0 \leq k < w(P) = \sum_{p \in P} w(p)$. Then there are at most \smash{$\frac{w(P) \cdot |Q|}{w(P) - k}$} shifts $s$ with $\sum_{p \in P : p + s \not\in Q} w(p) \leq k$.
\end{lemma}
\begin{proof}
Let $S$ denote the set of shifts $s$ with $\sum_{p \in P : p + s \not\in Q} w(p) \leq k$. Let $x = \One_Q$ be the indicator vector of $Q$ and let $y$ be the vector defined by $y[N - p] = w(p)$ if $p \in P$ and $y[N - p] = 0$ otherwise. Since for any integer shift $s$ we have
\begin{equation*}
    (x \conv y)[N + s] = \sum_{\substack{i, j\\i + j = N + s}} x[i] \cdot y[i] = \sum_{\substack{q \in Q\\p \in P\\q + N - p = N + s}} w(p) = \sum_{\substack{p \in P\\q \in Q\\p + s = q}} w(p),
\end{equation*}
it holds that $s \in S$ if and only if $(x \conv y)[N + s] \geq w(P) - k$. Thus, $\sum_{s \in S} (x \conv y)[N + s] \geq |S| \cdot (w(P) - k)$. On the other hand, it holds that $\sum_{s \in \Int} (x \conv y)[N + s] \leq w(P) \cdot |Q|$. The claim follows by combining both bounds.
\end{proof}

\begin{lemma}[Constellation with Weighted Mismatches to Partial Convolution] \label{lem:constellation-mismatches-to-partial-conv}
If the Partial Convolution problem is in deterministic time $\Order((|A| + |C - B|) \cdot \log^c N)$, then the Constellation problem with weighted mismatches is in deterministic time $\Order(n k \log^{c+1} N)$, provided that $0 \leq k \leq (1 - \Omega(1)) w(P)$.
\end{lemma}
\begin{proof}
This reduction is similarly set up as in \cref{lem:constellation-to-partial-conv}, but requires some more care to deal with the mismatches. We will again solve the Constellation problem over $\Int / N \Int$, with up to $k$ weighted mismatches. That is, for two given sets~\makebox{$P, Q \subseteq \Int / N \Int$} and a weight function $w : P \to \Int_{> 0}$, our goal is to compute the set $S$ of shifts $s$ satisfying that
\begin{equation*}
    \sum_{\substack{p \in P\\p + s \not\in Q}} w(p) \leq k    
\end{equation*}
(where addition is modulo $N$).

As before, let $N' = \frac{N}{2}$, $P' = \set{p \bmod N' : p \in P}$ and $Q' = \set{q \bmod N' : q \in Q}$. The weights remain unchanged, except that whenever two elements from $P$ collide under the modulo operation, we add their weights: $w'(p') = w(p') + w(p' + N')$. We solve the instance~$(P', Q', w', k)$ recursively. Note that also for the recursive call we guarantee the technical condition $0 \leq k \leq (1 - \Omega(1)) w'(P')$ since $w(P) = w'(P')$. As the result of the recursive call, we have computed the set $S' \subseteq \Int / N' \Int$ of shifts $s'$ satisfying $\sum_{p' \in P' : p' + s' \not\in Q'} w'(p') \leq k$. Next, we set up the Partial Convolution instance
\begin{align*}
    A &= Q \subseteq \Int / N \Int, \\
    B &= -P \subseteq \Int / N \Int, \\
    C &= S' + \set{0, N'} \subseteq \Int / N \Int,
\end{align*}
where $x$ is the indicator vector of $Q = A$, and $y$ is the vector defined by $y[-p] = w(p)$ if $p \in P$ and~\makebox{$y[-p] = 0$} otherwise (where negation is modulo $N'$). Note that indeed $\supp(x) = A$ and $\supp(y) = B$. We solve this instance using the Partial Convolution oracle, and return
\begin{equation*}
    S = \set{s \in C : (x \conv_N y)[s] \geq w(P) - k}.
\end{equation*}

\begin{claim}[Correctness] \label{lem:constellation-mismatches-to-partial-conv:clm:correctness}
$S = \set{s \in \Int / N \Int : \sum_{p \in P : p + s \not\in Q} w(p) \leq k}$.
\end{claim}
\begin{proof}
Suppose that $s$ satisfies $\sum_{p \in P : p + s \not\in Q} w(p) \leq k$. First, note that the condition $p + s \in Q$ implies that $p' + s' \in Q'$ where $p' = p \bmod N'$ and $s' = s \bmod N'$. In particular, we have that $\sum_{p' \in P' : p' + s' \not\in Q'} w'(p') \leq \sum_{p \in P : p + s \not\in Q} w(p) \leq k$. Assuming that the recursive call is correct, it follows that $s' \in S'$ and thus $s \in C$. It remains to observe that $(x \conv_N y)[s] = \sum_{p \in P : p + s \in Q} w(p)$ for all $s \in C$.
\end{proof}

\begin{claim}[Running Time] \label{lem:constellation-mismatches-to-partial-conv:clm:time}
The algorithm runs in time $\Order(n k \log^{c+1} N)$.
\end{claim}
\begin{proof}
A single call to the Partial Convolution oracle takes time $\Order((|A| + |C - B|) \cdot \log^c N)$. Recalling, as in \cref{lem:constellation-to-partial-conv:clm:time}, that $C = S' + \set{0, N'}$ and $B = -P = -(P' + \set{0, N'})$, we have that
\begin{equation*}
    C - B = S' + P' + \set{0, N', 2N'}.
\end{equation*}
To bound the size of $S' + P'$, we use that
\begin{equation*}
    |S' + P'| \leq |Q'| + \sum_{s \in S'} |(P' + s) \setminus Q'| \leq |Q'| + |S'| \cdot k.
\end{equation*}
Finally, by \cref{lem:constellation-mismatches-number} and the assumption that $k \leq (1 - \Omega(1)) w(P)$, the size of $S'$ is bounded by
\begin{equation*}
    |S'| \leq \frac{w(P') \cdot |Q'|}{w(P') - k} \leq \frac{w(P') \cdot |Q'|}{\Omega(w(P'))} = \Order(|Q'|).
\end{equation*}
By combining the previous three bounds, we obtain that $|C - B| = \Order(|Q'| + |S'| \cdot k) = \Order(|Q'| \cdot k) = \Order(n k)$. Thus, a single call to the oracle takes time $\Order(n k \log^c N)$. Moreover, the recursion depth of the algorithm is bounded by $\Order(\log N)$, and therefore the total running time is $\Order(n k \log^{c+1} N)$ as claimed.
\end{proof}

\noindent
In combination, both claims complete the proof of \cref{lem:constellation-mismatches-to-partial-conv}.
\end{proof}
\section{Partial Convolution and the Baur-Strassen Theorem} \label{sec:conv}
In this section we design our algorithm for the Partial Convolution problem, thereby proving the following \cref{thm:partial-conv}.

\begin{restatable}[Partial Convolution]{theorem}{thmpartialconv} \label{thm:partial-conv}
Let $A, B, C \subseteq [N]$ and $x, y \in \Int^N$ with $\supp(x) \subseteq A$, $\supp(y) \subseteq B$. Then we can compute $(x \conv y)[c]$ for all $c \in C$ in deterministic time $\Order((|A| + |C - B|) \log^5 (N \Delta) \polyloglog(N \Delta))$, where $\Delta$ is the largest entry in $x, y$ in absolute value.
\end{restatable}

The algorithm is algebraic in nature and requires some background on algebraic complexity. To this end, we start with a detailed description of arithmetic circuits in \cref{sec:conv:sec:circuits}, and then progressively develop the algorithm in \cref{sec:conv:sec:sparse-conv,sec:conv:sec:large-order,sec:conv:sec:partial-conv}.

\subsection{Arithmetic Circuits} \label{sec:conv:sec:circuits}
The basic algebraic model of computation that we are going to work with is \emph{arithmetic circuits}. While we are ultimately interested in algorithms in the RAM model, it is necessary to phrase major parts of our algorithm in terms of arithmetic circuits so that we can make use of the Baur-Strassen theorem (which only applies in this restricted model). We start with the basic definitions.

\begin{definition}[Arithmetic Circuit]
An \emph{arithmetic circuit $C$} over the field $\Field$ and the variables $x_1, \dots, x_n$ is a directed acyclic graph as follows. The nodes are called \emph{gates}, and are of the following two types: Each gate either has in-degree $0$ and is labeled with a variable $x_i$ or a constant $\alpha \in \Field$, or it has in-degree $2$ and is labeled with an arithmetic operation (${+}, {-}, {\times}, {/}$).
\end{definition}

We refer to the gates labeled by variables $X_i$ as \emph{input gates}, by constants as \emph{constant gates} and by an operation ${\circ} \in \set{{+}, {-}, {\times}, {/}}$ as \emph{$\circ$-gates}. A gate with out-degree $0$ is called an \emph{output gate.} The \emph{size} of an arithmetic circuit $C$, denoted by $|C|$, is the number of gates plus number of edges in $C$. Note that each gate in an arithmetic circuit computes a rational function $f(x_1, \dots, x_n)$ in a natural way: Input gates compute~$x_i$, constant gates compute the constant function $\alpha$, and each, say, $\times$-gate computes the product of the functions computed by its two incoming gates. We say that an arithmetic circuit computes functions~\makebox{$f_1, \dots, f_m$} if there are $m$ output gates computing these respective functions.

\paragraph{Partial Derivatives and the Baur-Strassen Theorem}
To a rational function $f(x_1, \dots, x_n)$ over an arbitrary field, we can naturally associate the formal partial derivatives \smash{$\frac{\partial f}{\partial x_i}(x_1, \dots, x_n)$} defined by the basic derivative rules ($(f + g)' = f' + g'$, $(f g)' = f' g + f g'$ and $(f / g)' = (f' g - f g') / g^2$). The Baur-Strassen theorem provides a way to efficiently compute all partial derivatives of a function computed by an arithmetic circuit:

\thmbaurstrassen*

\paragraph{Arithmetic Circuits for Transposed Vandermonde Matrices}
Another ingredient to our algorithm is the following well-known theorem about the algebraic complexity of evaluating matrix-vector products of \emph{transposed Vandermonde matrices:}

\begin{lemma}[Transposed Vandermonde Systems] \label{lem:transposed-vandermonde}
Let $\Field$ be a field, let $a_1, \dots, a_n \in \Field$ be pairwise distinct and consider the transposed Vandermonde matrix
\begin{equation*}
	V = 
	\begin{bmatrix}
		1 & 1 & \cdots & 1 \\
		a_1 & a_2 & \cdots & a_n \\
		a_1^2 & a_2^2 & \cdots & a_n^2 \\
		\vdots & \vdots & \ddots & \vdots \\
		a_1^{n-1} & a_2^{n-1} & \cdots & a_n^{n-1}
	\end{bmatrix}.
\end{equation*}
Then $V$ has full rank, and there is are arithmetic circuits with inputs $x = (x_1, \dots, x_n)$ computing the linear functions $V x$ and $V^{-1} x$. Both circuits have size $\Order(n \log^2 n)$ can be constructed in time $\Order(n \log^2 n)$.
\end{lemma}

This lemma can be proven in several ways, see e.g.~\cite{KaltofenL88,Li00,Pan01}. The perhaps simplest way is to observe that the same statement for non-transposed Vandermonde matrices is better known as multipoint evaluation and interpolation of univariate polynomials, which is a textbook result~\cite{vonzurGathenG13}. Moreover, by the \emph{transposition principle}~\cite[Theorem~3.4.1]{Pan01} it is known that the complexities of evaluating a linear map and its transpose are asymptotically the same.

\subsection{Arithmetic Circuits for Sparse Convolution} \label{sec:conv:sec:sparse-conv}
The following lemma is implicit in \cite{BringmannFN22}, but since in their paper Bringmann et al.\ were not forced to explicitly construct arithmetic circuits, we quickly repeat the construction.

\begin{lemma}[Sparse Convolution Circuit] \label{lem:sparse-conv-circuit}
Let $A, B, T \subseteq [N]$ such that $A + B \subseteq T$, let $\Field$ be a field and let~$\omega \in \Field$ have multiplicative order at least~$N$. Then we can construct an arithmetic circuit over $\Field$ with inputs $x_a$ (for~$a \in A$) and $y_b$ (for~$b \in B$), and outputs $z_c$ (for $c \in T$) defined by
\begin{equation*}
	z_c = \sum_{\substack{a \in A\\b \in B\\a + b = c}} x_a \cdot y_b.
\end{equation*}
The circuit has size $\Order(|T| \log^2 |T|)$, and it takes time $\Order(|T| \log^2 |T| + |T| \log N)$ to construct the circuit.
\end{lemma}
\begin{proof}
Let us assume that $A, B \subseteq T$; otherwise simply include the missing elements which does not increase the size of $T$ asymptotically. We write $T = \set{c_1, \dots, c_t}$ where $t = |T|$. The circuit is constructed in several steps. Our first step is to precompute all powers $\omega^{c_1}, \dots, \omega^{c_t}$. Note that these powers are pairwise distinct since $\omega$ has multiplicative order at least~$N$ and $c_1, \dots, c_t < N$. Writing
\begin{equation*}
	V = 
	\begin{bmatrix}
		1 & 1 & \cdots & 1 \\
		\omega^{c_1} & \omega^{c_2} & \cdots & \omega^{c_t} \\
		\vdots & \vdots & \ddots & \vdots \\
		\omega^{(t-1)c_1} & \omega^{(t-1)c_2} & \cdots & \omega^{(t-1)c_t}
	\end{bmatrix},
\end{equation*}
by \cref{lem:transposed-vandermonde} we can compute arithmetic circuits that respectively compute the matrix-vector products with $V$ and $V^{-1}$. Using these circuits, our first step is to construct circuits computing the values~$\widehat x_i$ and~$\widehat y_i$~(for~$i \in [t]$) defined by
\begin{equation*}
	\widehat x_i = \sum_{a \in A} \omega^{ia} \cdot x_a, \quad \widehat y_j = \sum_{b \in B} \omega^{jb} \cdot y_b.
\end{equation*}
Indeed, these circuits follow immediately, since $\widehat x = (\widehat x_0, \dots, \widehat x_{t-1})$ is the result of the matrix-vector product~$V$ times the length-$t$ vector with $i$-th entry $x_{c_i}$ if $c_i \in A$ and $0$ otherwise (and similarly for~$\widehat y$). Next, we compute for all $i \in [t]$ the values
\begin{equation*}
	\widehat z_i = \widehat x_i \cdot \widehat y_i = \parens*{\sum_{a \in A} \omega^{ia} \cdot x_a} \parens*{\sum_{b \in B} \omega^{ib} \cdot y_b} = \sum_{\substack{a \in A\\b \in B}} \omega^{i(a + b)} x_a \cdot y_b = \sum_{c \in T} \omega^{ic} \cdot z_c,
\end{equation*}
where $z_c$ is as in the lemma statement. Observe that similarly to before we can express $\widehat z = (\widehat z_0, \dots, \widehat z_{t-1})$ as the matrix-vector product $V \cdot (z_{c_1}, \dots, z_{c_t})^T$. It follows that $(z_{c_1}, \dots, z_{c_t})$ can be computed via the matrix-vector product $V^{-1} \widehat z^T$ for which we also have an arithmetic circuit available. This completes the description of the sparse convolution circuit.

It was constructed by appropriately composing three copies of the circuits constructed in \cref{lem:transposed-vandermonde} all of which have size $\Order(t \log^2 t)$ and that can be constructed in the same time. The only other contribution to the running time was the precomputation of the powers $\omega^{c_1}, \dots, \omega^{c_t}$ that takes $\Order(t \log N)$ by repeated squaring.
\end{proof}

\subsection{Finding Large-Order Elements} \label{sec:conv:sec:large-order}
A crucial ingredient of \cref{lem:sparse-conv-circuit} is that we need to provide a field element $\omega$ with large multiplicative order. It is known that in any finite field $\Field_p$, an almost-constant fraction of elements is primitive (i.e., has maximum multiplicative order $p-1$), and therefore we can find an element with large order in randomized time \smash{$\widetilde\Order(1)$} by sampling. Unfortunately, the deterministic construction of such elements is known to be a notoriously hard problem. Even worse: For the field $\Field_p$ to contain an element with multiplicative order at least~$N$, we have to find a prime $p$ of size at least~$N$. Again, this is in polylogarithmic randomized time, but the best deterministic algorithm takes time $N^{\frac12+\order(1)}$~\cite{LagariasO87,TaoCH12}.

Luckily, a line of research has investigated how to construct large-order elements in finite fields of prime power order $q = p^d$ with comparably small characteristic $p$; see e.g.~\cite{Shoup90,Shparlinski96,Gao99,Cheng05,Cheng07}. Before we go into more details, let us first recap some basics on how to compute with finite fields in the RAM model.

\paragraph{Finite Field Arithmetic}
Let $q = p^d$ be a prime power. Recall that the finite field $\Field_p$ can be represented as $\Int / p \Int$, the integer modulo $p$. The finite field $\Field_q$ can be represented as $\Field_p[X] / \ideal f$ where $f \in \Field_p[X]$ is an irreducible polynomial of degree $d$. In the word RAM model with words storing numbers up to $p$, we can thus represent field elements from $\Field_q$ as length-$d$ lists of coefficients from $\Field_p$. Using basic polynomial arithmetic, we can compute the field operations (${+}, {-}, {\times}, {/}$) of~$\Field_q$ in time \smash{$\widetilde\Order(\log q)$}.

\paragraph{Constructing Large-Order Elements}
Almost any of the results~\cite{Shoup90,Shparlinski96,Gao99,Cheng05,Cheng07} is sufficient for our purposes. We will use a result of Cheng~\cite{Cheng07}, which gives the best running time in terms of lower-order factors $\log N$. See also \cite[Lemmas~8,~9]{BringmannFN22} for a compact proof.

\begin{lemma}[Constructing Large-Order Elements~\cite{Cheng07}] \label{lem:large-order}
Let $p \geq 7$ be a prime. In time $\widetilde\Order(p)$ we can construct an irreducible polynomial $f \in \Field_p[X]$ of degree $p-1$, and an element $\omega \in \Field_p[X] / \ideal f = \Field_{p^{p-1}}$ of order at least $2^p$.
\end{lemma}

\subsection{Partial Convolution} \label{sec:conv:sec:partial-conv}
We are finally ready to give our algorithm for the Partial Convolution problem. Recall that in this problem, we consider three sets $A, B, C \subseteq [N]$ and two vectors $x, y \in \Int^N$ with $\supp(x) \subseteq A$ and $\supp(y) \subseteq B$. The task is to compute the convolution vector $x \conv y$ \emph{restricted} to the positions in $C$. 

We start to solve a relaxation of this problem; see the following \cref{lem:partial-conv-modulo-p}. Here the goal is to compute the numbers \emph{modulo $p$}, for some prime $p$, and we can additionally assume that a small superset $T$ of $A + B$ is known. We remove these assumptions later in \cref{lem:partial-conv-superset,thm:partial-conv}.

\begin{lemma}[Partial Convolution modulo Prime] \label{lem:partial-conv-modulo-p}
Let $A, B, C, T \subseteq [N]$ and $x, y \in \Int^N$ with $\supp(x) \subseteq A$, $\supp(y) \subseteq B$ and $C - B \subseteq T$. Let $p > \log N$ be a prime. Then we can compute $(x \conv y)[c] \bmod p$ for all $c \in C$ in deterministic time $\Order((|A| + |T|) \log^2(N) \cdot p \polylog(p))$.
\end{lemma}
\begin{proof}
As a first step, we set up an appropriate finite field containing an element of large multiplicative order. To this end, we apply \cref{lem:large-order} to find an irreducible polynomial $f \in \Field_p[X]$ of degree $p - 1$, and an element $\omega \in \Field_p[X] / \ideal f$ whose multiplicative order is guaranteed to be at least $2^p$. Let us represent the finite field $\Field_{p^{p-1}}$ by $\Field_p[X] / \ideal f$; then we have access to an element $\omega \in \Field_{p^{p-1}}$ with order at least $2^p > N$.

Next, we will use $\omega$ to construct a sparse convolution circuit over $\Field_{p^{p-1}}$. Applying \cref{lem:sparse-conv-circuit} with~$-B$,~$C$ and $T$ (note that we indeed have $-B + C = C - B \subseteq T$) yields an arithmetic circuit with inputs $y_b$ (for~$b \in B$) and~$z_c$~(for~$c \in C$), and outputs $w_a$ (for~$a \in T$) computing the bilinear functions
\begin{equation*}
	w_a = \sum_{\substack{b \in B\\c \in C\\-b + c = a}} y_b \cdot z_c = \sum_{\substack{b \in B\\c \in C\\a + b = c}} y_b \cdot z_c.
\end{equation*}
We will modify this circuit in several steps. First, we add new input gates $x_a$ (for $a \in A$) to the circuit. Then, by adding the appropriate wires and a new output gate $w$, we let the circuit compute
\begin{equation*}
	w = \sum_{a \in A \cap T} x_a \cdot w_a = \sum_{\substack{a \in A \cap T\\b \in B\\c \in C\\a + b = c}} x_a \cdot y_b \cdot z_c = \sum_{\substack{a \in A\\b \in B\\c \in C\\a + b = c}} x_a \cdot y_b \cdot z_c.
\end{equation*}
We delete all output gates other than $w$, so that the circuit computes the single output $w$. This allows us to apply the Baur-Strassen theorem to construct an augmented circuit that computes the partial derivatives~$\frac{\partial w}{\partial z_c}$ for all $c \in C$. The description of the arithmetic circuit is complete. Note that it computes
\begin{equation*}
	\frac{\partial w}{\partial z_c} = \frac{\partial}{\partial z_c} \sum_{\substack{a \in A\\b \in B\\c \in C\\a + b = c}} x_a \cdot y_b \cdot z_c = \sum_{\substack{a \in A\\b \in B\\a + b = c}} x_a \cdot y_b. 
\end{equation*}
This is exactly the convolution function that we set out to compute. By plugging in the vector $x$ and $y$ into the same-named inputs $x_a$ and $y_b$, we can therefore read of the convolution $(x \conv y)[c]$ from the outputs $\frac{\partial w}{\partial z_c}$. Recall that all computations are modulo $p$, and thus we in fact only have access to $(x \conv y)[c] \bmod p$.

Let us finally analyze the running time. Constructing the finite field $\Field_{p^{p-1}}$ as well as the large-order element $\omega$ takes time \smash{$\widetilde\Order(p)$}, which is negligible. The construction of the sparse convolution circuit takes time~$\Order(|T| \log^2 N)$ and leads to a circuit of size~$\Order(|T| \log^2 N)$. All subsequent modifications, including adding outputs as well as applying the Baur-Strassen \cref{thm:baur-strassen}, do not increase the size of the circuit asymptotically and run in linear time in the circuit size, $\Order(|T| \log^2 N + |A|)$. Finally, evaluating the circuit takes time $\Order(|T| \log^2 (N) + |A|)$ plus $\Order(|T| \log^2(N) + |A|)$ field operations over $\Field_{p^{p-1}}$. Recall that each field operation over a finite field \smash{$\Field_{p^d}$} takes time \smash{$d \polylog(d)$}, and therefore the total running time is bounded by~\smash{$\Order((|A| + |T|) \log^2(N) \cdot p \polylog(p))$} as claimed. 
\end{proof}

We will now remove the restriction of the previous lemma that it computes the outputs modulo some prime $p$, by applying it repeatedly for several primes $p$, and by using the Chinese Remainder Theorem.

\begin{lemma}[Partial Convolution with Superset] \label{lem:partial-conv-superset}
Let $A, B, C, T \subseteq [N]$ and $x, y \in \Int^N$ with \makebox{$\supp(x) \subseteq A$}, \makebox{$\supp(y) \subseteq B$} and~\makebox{$C - B \subseteq T$}. Then we can compute $(x \conv y)[c]$ for all $c \in C$ in deterministic time $\Order((|A| + |T|) \log^4 (N \Delta) \polyloglog(N \Delta))$, where $\Delta$ is the largest entry in $x, y$ in absolute value.
\end{lemma}
\begin{proof}
First, precompute a list of $r = \ceil{\log (2N\Delta)}$ distinct primes $p_1, \dots, p_r$ larger than $\log N$. By the Prime Number Theorem, we can pick primes $p_1, \dots, p_r$ of size at most $\Order(\log(N \Delta) \log\log (N \Delta))$. Moreover, we can find these primes in time $\Order(\log (N \Delta) \polyloglog(N \Delta))$ using for instance the Sieve of Eratosthenes.

We apply \cref{lem:partial-conv-modulo-p} $r$ times, for the sets $A, B, C, T$ and the primes $p = p_1, \dots, p_r$ respectively. We thereby compute, for each $c \in C$, the values $(x \conv y)[c] \bmod p_i$ (for $i = 1, \dots, r$). Since $(x \conv y)[c]$ is at most~$\Delta N$ in absolute value, by the Chinese Remainder Theorem these modular equations uniquely determine the integer value $(x \conv y)[c]$.

Finally, we analyze the running time.
Repeatedly calling \cref{lem:partial-conv-modulo-p} for $r = \Order(\log (N \Delta))$ times takes time $\Order((|A| + |T|) \log^4(N \Delta) \polyloglog(N \Delta))$. Afterwards, solving the congruences with the Chinese Remainder Theorem takes time \smash{$\Order(\log^2 (\prod_{i=1}^r p_i)) = \Order(\log^2(N) \polyloglog(N))$}~\cite{vonzurGathenG13} for each element $c \in C$, and thus negligible time~\makebox{$\Order(|C| \log^2(N \Delta) \polyloglog(N \Delta))$} in total. 
\end{proof}

Finally, we can remove the assumption that a small superset $T$ of $C - B$ is known, simply by precomputing the set $T = C - B$ exactly. To this end, we use the deterministic sparse convolution algorithm by Bringmann, Fischer and Nakos~\cite{BringmannFN22} that runs in time $\Order(|C - B| \log^5(N) \polyloglog(N))$. All in all, we obtain the claimed \cref{thm:partial-conv}.

\subsection{Corollary for 3SUM}
As a direct consequence of \cref{thm:partial-conv} we obtain an algorithm for the \#AllNumbers3SUM problem, running in time $\widetilde\Order(n + \min\set{|A + B|, |A + C|, |B + C|})$ (see \cref{thm:all-numbers-3sum}).

\begin{proof}[Proof of \cref{thm:all-numbers-3sum}]
It suffices to solve the problem in time $\widetilde\Order(|A + B| + |C|)$; the other running times can be obtained by exchanging the roles of $A, B, C$. The output of the problem consists of the following three types of queries:
\begin{enumerate}
\item For all $a \in A$, compute $\abs{\set{(b, c) \in B \times C : -c - b = a}}$.
\item For all $b \in B$, compute $\abs{\set{(a, c) \in A \times C : -c - a = b}}$.
\item For all $c \in C$, compute $\abs{\set{(a, b) \in A \times B : a + b = -c}}$.
\end{enumerate}
It is known how to solve the third type using sparse convolution algorithms~\cite{BringmannFN22}: In time $\widetilde\Order(|A + B|)$ we can compute the sumset $A + B$ with multiplicities. Afterwards we can read off, for each $c \in C$, the multiplicity of $-c$ in the sumset $A + B$.

The first and second types require our new technology. Specifically, for the queries of the first type, we apply \cref{thm:partial-conv} with $A' = -C$, $B' = -B$ and $C' = A$, and let $x = \One_{-C}$ and $y = \One_{-B}$ be the indicator vectors of $A' = -C$ and $B' = -B$, respectively. (Here, for simplicity we assume that vectors can be indexed by negative numbers. This requirement can easily be avoided by shifting all sets by a common large offset.) The output consists of $(x \conv y)[a]$, for all $a \in C' = A$. By construction, we have
\begin{equation*}
	(x \conv y)[a] = \sum_{\substack{i, j \in \Int\\i + j = a}} \One_{-C}(i) \cdot \One_{-B}(j) = \abs{\set{(b, c) \in B \times C : -c - b = a}},
\end{equation*}
which is exactly as desired. The running time is bounded by $\widetilde\Order(|A'| + |C' - B'|) = \widetilde\Order(|C| + |A + B|)$. Finally, the queries of the second type can be computed analogously (by exchanging $A$ and $B$).
\end{proof}
\section{Corollaries for Point and String Pattern Matching} \label{sec:corollaries}
By combining the previous two sections (\cref{lem:constellation-to-partial-conv,lem:constellation-mismatches-to-partial-conv} with \cref{thm:partial-conv}), we have completed the proofs of our two Main \cref{thm:constellation,thm:constellation-mismatches}. Moreover, since \cref{lem:constellation-mismatches-to-partial-conv} works in the presence of \emph{weighted} mismatches, we obtain the following result which summarizes our Main Theorems in their strongest form:

\begin{theorem}[Deterministic Constellation with Mismatches] \label{thm:constellation-weights}
Given sets $P, Q \subseteq [N]$ of size at~most~$n$, positive integer weights $w : P \to \Int_{> 0}$ and $0 \leq k \leq (1 - \Omega(1))w(P)$, we can list all shifts $s$ with $\sum_{p \in P : p + s \not\in Q} w(p) \leq k$ in deterministic time $\Order(n k \polylog(N))$.
\end{theorem}

In the remainder of this section, we describe how to obtain \crefrange{cor:ppm-translations-exact}{cor:sparse-wildcard-matching}.

\subsection{Preliminaries on Point Pattern Matching} \label{sec:corollaries:sec:ppm}
We start with a systematic definition of the Point Pattern Matching problems that we consider. Consider the $d$-dimensional Euclidian space, and let $\norm{\cdot}$ denote the Euclidian norm. For two finite sets $P, Q \subseteq \Real^d$, the \emph{directed Hausdorff distance} from $P$ to $Q$ is defined as
\begin{equation*}
    h(P, Q) = \max_{p \in P} \min_{q \in Q} \norm{p - q},
\end{equation*}
and the \emph{undirected Hausdorff distance} between $P$ and $Q$ is defined as
\begin{equation*}
    H(P, Q) = \max\set{h(P, Q), h(Q, P)}.
\end{equation*}
Throughout we fix a group of \emph{transformations}~$T$. We consider two types of transformations: translations (in which case $T$ is the additive group $\Real^d$) and rigid motions---that is, translations, rotations and possibly reflections (in this case $T$ is the Euclidian group). For a transformation~$t \in T$, we write $t(p)$ to denote image under the transformation $t$ of some point $p \in \Real^d$, and for a set of points~$P$ we write~$t(P) = \set{t(p) : p \in P}$. Each Point Pattern Matching problem is defined for a group of transformations~$T$, with the goal of returning all transformations that move a given set $P$ close to another set $Q$. The exact meaning of ``close'' is formalized in the following four problem variants.

\begin{problem}[Exact Point Pattern Matching] \label{prob:ppm-exact}
Let $T$ be a group of transformations. Given point sets $P, Q \subseteq \Int^d$, return $S = \set{t \in T : t(P) \subseteq Q}$.
\end{problem}

\begin{problem}[Approximate Point Pattern Matching] \label{prob:ppm-approx}
Let $T$ be a group of transformations. Given point sets $P, Q \subseteq \Real^d$ and parameters $\epsilon \in (0, \frac12), \delta > 0$, return a set $S \subseteq T$ with:
\begin{itemize}[topsep=\smallskipamount, itemsep=0pt]
    \item\emph{Soundness:} For all $s \in S$, $h(s(P), Q) \leq (1 + \epsilon)\delta$.
    \item\emph{Completeness:} For all $t \in T$ with $h(t(P), Q) \leq \delta$, there is some $s \in S$ with $h(s(P), t(P)) \leq \epsilon \delta$.
\end{itemize}
\end{problem}

\begin{problem}[Exact Point Pattern Matching with Mismatches] \label{prob:ppm-exact-mismatches}
Let $T$ be a group of transformations. Given point sets $P, Q \subseteq \Int^d$ and a parameter $0 \leq k < |P|$, return $S = \set{t \in T : |t(P) \setminus Q| \leq k}$.
\end{problem}

\begin{problem}[Approximate Point Pattern Matching with Mismatches] \label{prob:ppm-approx-mismatches}
Let $T$ be a group of transformations. Given point sets $P, Q \subseteq \Real^d$ and parameters~\makebox{$0 \leq k < |P|$} and $\epsilon \in (0, \frac12), \delta > 0$, return a set $S \subseteq T$ with:
\begin{itemize}[topsep=\smallskipamount, itemsep=0pt]
    \item\emph{Soundness:} For all $s \in S$, there is some $P_0 \subseteq P$ with $|P_0| \geq |P| - k$ and $h(s(P_0), Q) \leq (1 + \epsilon)\delta$.
    \item\emph{Completeness:} For all $t \in T$ for which there is some set $P_0 \subseteq P$ with $|P_0| \geq |P| - k$ and $h(t(P_0), Q) \leq \delta$, there is some $s \in S$ with $h(s(P_0), t(P_0)) \leq \epsilon \delta$.
\end{itemize}
\end{problem}

As before we assume that the point sets $P$ and $Q$ have size at most $n$. Moreover, we let~$N$ denote the boundary of $P, Q$ in the following sense. For integer points $P, Q \subseteq \Int^d$ we let $N$ denote the coordinate-wise maximum of all points in $P$ and $Q$ (in absolute value). For real points~\makebox{$P, Q \subseteq \Real^d$} we assume that all coordinates of all points are available up to some fixed precision using $\log N$ bits. In particular, we assume that the maximum distance between any two points is~$N$ and that the minimum distance between any two points is $N^{-1}$.

For the exact Point Pattern Matching problem with mismatches (\cref{prob:ppm-exact-mismatches}), we will also consider a weighted version where, as before, we report $S = \set{t \in T : \sum_{p \in P : t(p) \not\in Q} w(p) \leq k}$. We need this problem as an intermediate step in the upcoming \cref{lem:ppm-approx-mismatches-to-constallation}.

\subsection{From Point Pattern Matching to Constellation}
We now state three reductions from Point Pattern Matching to the Constellation problem (\cref{lem:ppm-exact-mismatches-to-constallation,lem:ppm-approx-mismatches-to-constallation,lem:ppm-approx-rigid-to-constallation}). While these reductions are in spirit due to Cardoze and Schulman~\cite{CardozeS98}, they at times relied on randomization. Fortunately, the use of randomization is not essential here and can be easily removed. For completeness, we include proofs for the two reductions that cannot be directly reused from~\cite{CardozeS98} (\cref{lem:ppm-exact-mismatches-to-constallation,lem:ppm-approx-mismatches-to-constallation}), and only treat the one reduction that is originally deterministic as a black-box (\cref{lem:ppm-approx-rigid-to-constallation}).

\begin{lemma}[Exact Translations with Mismatches to Constellation] \label{lem:ppm-exact-mismatches-to-constallation}
If the Constellation problem with~$k$ (weighted) mismatches is in deterministic time $f(n, k, N)$, then the exact Point Pattern Matching problem with translations and $k$ (weighted) mismatches is in deterministic time $\Order(f(n, k, (4N)^d))$.
\end{lemma}

We remark that in \cite{CardozeS98}, this reduction is achieved by mapping each high-dimensional point to an integer via a random linear combination of its coordinates. It is easy to prove that this preserves solutions and does not introduce false positives with good probability. Our work-around is standard: We instead encode all coordinates into one big integer.

\begin{proof}
Let us start with the unweighted version, and assume without loss of generality that all coordinates are nonnegative (by shifting the sets $P, Q$ if necessary, which increases $N$ to at most~$2N$). The idea is to encode points $p = (p_0, \dots, p_{d-1}) \in [2N]^d$ as integers \smash{$\overline p = \sum_{i=0}^{d-1} p_i \cdot (4N)^i$}. Construct the sets $\overline P = \set{\overline p : p \in P}$ and $\overline Q = \set{\overline q : q \in Q}$, view $(\overline P, \overline Q)$ as an instance of the Constellation problem and solve it using the oracle. The solution consists of several shifts~$s'$ with~\smash{$|(\overline P + s') \setminus \overline Q| \leq k$}. We say that such a shift~$s'$ is \emph{admissible} if we can express $s$ as \makebox{$s = \sum_{i=0}^{d-1} s_i \cdot (4N)^i$}, for integers~\makebox{$s_0, \dots, s_{d-1} \in [2N]$}. For each admissible shift $s'$, we report $(s_0, \dots, s_{d-1}) \in [2N]^d$ as a solution.

For the correctness argument, we first suppose that $s = (s_0, \dots, s_{d-1}) \in [2N]^d$ satisfies that $|(P + s) \setminus Q| \leq k$. Then it is easy to check that $|(\overline P + \overline s) \setminus \overline Q| \leq k$ and that $\overline s$ is admissible. Hence, the algorithm reports $s$. For the other direction, suppose that the algorithm reports $s = (s_0, \dots, s_{d-1}) \in [2N]^d$. It does so only whenever~$\overline s$ satisfies $|(\overline P + \overline s) \setminus \overline Q| \leq k$. We claim that then also $|(P + s) \setminus Q| \leq k$. In fact, we prove the stronger statement that for any two points $p, q \in [2N]^d$, $\overline p + \overline s = \overline q$ implies that $p + s = q$. The proof is by induction on $d$, where the case $d = 1$ is trivial. For $d > 1$, we rewrite $\overline p + \overline s = \overline q$ as
\begin{equation*}
    \sum_{i=0}^{d-1} p_i \cdot (4N)^i + \sum_{i=0}^{d-1} s_i \cdot (4N)^i = \sum_{i=0}^{d-1} q_i \cdot (4N)^i.
\end{equation*}
Taking this equation modulo $4N$, we have that $p_0 + s_0 \equiv q_0 \mod{4N}$. Recall that $0 \leq p_0, s_0, q_0 < 2N$, and thus $p_0 + s_0 = q_0$. It follows that
\begin{equation*}
    \sum_{i=1}^{d-1} p_i \cdot (4N)^{i-1} + \sum_{i=1}^{d-1} s_i \cdot (4N)^{i-1} = \sum_{i=1}^{d-1} q_i \cdot (4N)^{i-1},
\end{equation*}
and we continue by induction to prove that $p_i + s_i = q_i$ for all $i = 1, \dots, d-1$. All in all, we have proven that $p + s = q$, as claimed.

Observe that the reduction is not affected by weights, as we construct a one-to-one mapping between $P$ and $\overline P$.

Let us finally comment on the running time. The most significant contribution is calling the Constellation oracle on $(\overline P, \overline Q)$, which takes time $f(n, k, (4N)^d)$. The pre- and post-processing runs in linear time $\Order(nd)$, which is dominated by $f(n, k, (4N)^d)$, simply to read the input.
\end{proof}

\begin{lemma}[Approximate Translations with Mismatches to Constellation] \label{lem:ppm-approx-mismatches-to-constallation}
If the Constellation problem with~up to~$k$ weighted mismatches is in deterministic time $f(n, k, N)$, then the approximate Point Pattern Matching problem with translations and up to~$k$ mismatches is in deterministic time $\Order(f(\epsilon^{-\Order(d)} n, k, \Order(d^{-1/2} \epsilon^{-1} \delta^{-1} N)^d))$.
\end{lemma}
\begin{proof}
Our strategy is to reduce the approximate case to the exact case and then apply the reduction from \cref{lem:ppm-exact-mismatches-to-constallation}. To this end, let \smash{$\alpha \in [\frac{\epsilon \delta}{8d^{1/2}}, \frac{\epsilon \delta}{4d^{1/2}}]$} be any number representable by $\Order(\log(d \epsilon^{-1} \delta^{-1}))$ bits. We construct the following two integer point sets $P', Q'$:
\begin{itemize}
    \item Let $P' \subseteq \Int^d$ be the point set obtained from by rounding $\alpha^{-1} P$ to the (approximately) closest integer points. We have that $H(\alpha^{-1} P, P') \leq d^{1/2}$ and therefore $H(P, \alpha P') \leq \alpha d^{1/2} \leq \frac{\epsilon \delta}{4}$.
    \item Let $Q' \subseteq \Int^d$ be the set of points that are in Euclidian distance at most $\alpha^{-1} (1 + \frac{\epsilon}{2})\delta$ to some point in $\alpha^{-1} Q$, plus possibly some points with distance at most $\alpha^{-1} (1 + \frac{3\epsilon}{4}) \delta$. (This slack is necessary since we are working with finite-precision arithmetic.) In particular, we have that $H(Q, \alpha Q') \leq (1 + \frac{3\epsilon}{4}) \delta$.
\end{itemize}
Moreover, we define the weight $w(p')$ of a point $p'$ to be the number of points in $P$ that collided into $p'$ under the rounding. We view $(P', Q', w, k)$ as an instance of exact Point Pattern Matching with translations and up to~$k$ weighted mismatches, and solve it using the reduction in \cref{lem:ppm-exact-mismatches-to-constallation}. The output is the set $S' = \set{s' \in \Int : \sum_{p' \in P' : p' + s' \not\in Q'} w(p') \leq k}$. Finally, we construct and report $S = \set{\alpha s' : s' \in S'}$.

Let us first argue that this output is correct (i.e., satisfies the conditions in \cref{prob:ppm-approx-mismatches}):
\begin{itemize}
    \item \emph{Soundness:} Any shift $s \in S$ has the form $\alpha s'$ where $s' \in S$. It follows that there is some subset~\makebox{$P_0' \subseteq P'$} with weight $w(P_0') \geq w(P') - k$ satisfying that $P_0' + s' \subseteq Q'$. Let $P_0 \subseteq P$ denote a pre-image of~$P_0'$ under the rounding. Since the weights count the number of points before the rounding, we have that $|P_0| \geq w(P') - k = |P| - k$. Moreover, the rounding guarantees that $H(P_0, \alpha P_0') \leq \alpha d^{1/2} \leq \frac{\epsilon\delta}{4}$. Therefore, using the triangle inequality we have that
    \begin{align*}
        h(P_0 + s, Q)
        &\leq H(P_0 + s, \alpha P_0' + \alpha s') + h(\alpha P_0' + \alpha s', \alpha Q') + H(\alpha Q', Q) \\
        &\leq H(P_0, \alpha P_0') + \alpha \cdot h(P_0' + s', Q') + H(\alpha Q', Q) \\
        &\leq \tfrac{\epsilon\delta}{4} + 0 + (1 + \tfrac{3\epsilon}{4}) \delta \\
        &= (1 + \epsilon) \delta.
    \end{align*}
    \item \emph{Completeness:} Suppose that there is a shift $t \in \Real^d$ and a subset $P_0 \subseteq P$ of size $|P_0| \geq |P| - k$ such that $h(P_0 + t, Q) \leq \delta$. Let $s'$ be the integer shift obtained by rounding $\alpha^{-1} t$, and let~$P_0'$ be the set obtained by rounding $\alpha^{-1} P_0$. By the rounding we have $h(\alpha P_0' + \alpha s', Q) \leq \delta + 2\alpha d^{1/2} \leq (1 + \frac{\epsilon}{2}) \delta$, and by the construction of $Q'$, we have that $P_0' + s' \subseteq Q'$. Since $w(P' \setminus P_0') = |P'| - |P_0'| \leq k$, it follows that $\sum_{p' \in P : p' + s' \not\in Q'} w(p') \leq k$. Thus~$s' \in S'$ and~\makebox{$s := \alpha s' \in S$}. Finally, we have that $h(P_0 + t, P_0 + s) \leq \alpha d^{1/2} \leq \epsilon \delta$ as desired.
\end{itemize}

It remains to analyze the running time. We have that $|P'| \leq |P|$ and $|Q'| \leq \epsilon^{-\Order(d)} |Q|$. Moreover, assuming that the real points in $P, Q$ are representable by $\log N$ bits, it follows that the integer points in $P', Q'$ have norm bounded by $\Order(\alpha^{-1} N) = \Order(d^{1/2} \epsilon^{-1} \delta^{-1} N)$. Therefore, assuming that $f(n, k, N)$ is the running time of Constellation with $k$ mismatches, the running time of \cref{lem:ppm-exact-mismatches-to-constallation} is $\Order(f(\epsilon^{-\Order(d)} n, k, \Order(d^{1/2} \epsilon^{-1} \delta^{-1} N)^d))$. The pre- and post-processing steps run in time $\epsilon^{-\Order(d)} n$ and are thus negligible.
\end{proof}

\begin{lemma}[Approximate Rigid Motions to Constellation, {{\cite[Section~6]{CardozeS98}}}] \label{lem:ppm-approx-rigid-to-constallation}
If the Constellation problem is in deterministic time $f(n, N)$ (where $f(\cdot, \cdot)$ is convex), then the approximate Point Pattern Matching problem with rigid motions is in deterministic time $\Order(n^{d-1} \cdot f(\epsilon^{-\Order(d)} n, \epsilon^{-1} \delta^{-1} 2^{\Order(d)} N))$.
\end{lemma}

\subsection{Corollaries for Point Pattern Matching}
For completeness, in this section we list how to derive \crefrange{cor:ppm-translations-exact}{cor:ppm-mismatches-approx} from \cref{thm:constellation,thm:constellation-mismatches} and the previous reductions.

\begin{proof}[Proof of \cref{cor:ppm-rigid-approx}]
We plug \cref{thm:constellation-weights} with $f(n, N) = \Order(n \polylog(N))$ into \cref{lem:ppm-approx-rigid-to-constallation}. The resulting running time is $\Order(n^{d-1} \cdot \epsilon^{-\Order(d)} n \polylog(\epsilon^{-1} \delta^{-1} 2^{\Order(d)} N)) = \Order(n^d \epsilon^{-\Order(d)} \polylog(\delta^{-1} N))$ (the $\poly(d)$ term is dominated by $\epsilon^{-\Order(d)}$).
\end{proof}

\begin{proof}[Proof of \cref{cor:ppm-mismatches-exact}]
We plug \cref{thm:constellation-weights} with $f(n, k, N) = \Order(n k \polylog(N))$ into \cref{lem:ppm-exact-mismatches-to-constallation}. The resulting running time is $\Order(n k \polylog((4N)^d)) = \Order(n k \poly(d \log(N)))$.
\end{proof}

\begin{proof}[Proof of \cref{cor:ppm-mismatches-approx}]
We plug \cref{thm:constellation-weights} with $f(n, k, N) = \Order(n k \polylog(N))$ into \cref{lem:ppm-approx-mismatches-to-constallation}. The resulting running time is $\Order(n k \epsilon^{-\Order(d)} \polylog((d \epsilon^{-1} \delta^{-1} N)^d)) = \Order(n k \epsilon^{-\Order(d)} \polylog(\delta^{-1} N))$ (the $\poly(d)$ term is dominated by $\epsilon^{-\Order(d)})$.
\end{proof}

\begin{proof}[Proof of \cref{cor:ppm-translations-exact}]
Apply \cref{cor:ppm-mismatches-exact} with $k = 0$.
\end{proof}

\begin{proof}[Proof of \cref{cor:ppm-translations-approx}]
Apply \cref{cor:ppm-mismatches-approx} with $k = 0$.
\end{proof}

\subsection{Corollary for Sparse Wildcard Matching}
Recall the Sparse Wildcard Matching problem: Let $\Sigma$ be an alphabet, let $0 \not\in \Sigma$ denote a designated character and let $* \not\in \Sigma$ denote a wildcard symbol. Let the \emph{text} $T \in (\Sigma \cup \set{0})^N$ and the \emph{pattern} $P \in (\Sigma \cup \set{*})^M$ be two strings (with $M \leq N$). We assume that $T$ contains at most $n$ non-zero characters and $P$ contains at most $n$ non-wildcard characters. The task is to report all substrings of $T$ that match $P$, where we understand that the wildcard~$*$ matches any character in $T$. As a corollary of \cref{thm:constellation}, this problem can be solved in time $\Order(n \polylog(N))$:

\begin{proof}[Proof of \cref{cor:sparse-wildcard-matching}]
Let $\Sigma(P)$ denote the non-wildcard characters appearing in the pattern string $P$. We construct sets~$P_\sigma = \set{i : P[i] = \sigma}$ and $Q_\sigma = \set{i : T[i] = \sigma}$ for each character $\sigma \in \Sigma(P)$. We then view each pair $(P_\sigma, Q_\sigma)$ as an instance of the Constellation problem, and compute the sets $S_\sigma$ of all shifts~$s$ with~\makebox{$P_\sigma + s \subseteq Q_\sigma$}. Finally, we report $S = \bigcap_\sigma S_\sigma$ (claiming that for each $s \in S$, $P$ matches $T[s\,..\,s + M)$).

It is easy to check that this algorithm is correct: Whenever $P$ matches $T[s\,..\,s+M)$, then for each non-wildcard character $\sigma$ we have that $P_\sigma + s \subseteq Q_\sigma$. Conversely, whenever $P_\sigma + s \subseteq Q_\sigma$ holds for all $\sigma$, then all non-wildcard characters in $P$ match $T[s\,..\,s+M)$. Since the wildcards match by definition, $P$ entirely matches $T[s\,..\,s+M)$.

By \cref{thm:constellation}, the total running time is bounded by \smash{$\Order(\sum_{\sigma \in \Sigma(P)} (|P_\sigma| + |Q_\sigma|) \polylog(N))$}. Since the total number non-wildcard characters in $P$ is at most $n$, we have that \smash{$\sum_{\sigma \in \Sigma(P)} |P_\sigma| \leq n$}. Similarly, since $T$ has at most $n$ non-zero characters and since~$0 \not\in \Sigma(P)$, we have \smash{$\sum_{\sigma \in \Sigma(P)} |Q_\sigma| \leq n$}. It follows that the total running time is $\Order(n \polylog(N))$.
\end{proof}

\section*{Acknowledgements}
I would like to thank Amir Abboud, Orr Fischer, Leo Wennmann and several anonymous reviewers for many helpful comments on an earlier version of this paper.
\bibliographystyle{plainurl}
\bibliography{references}

\begin{thebibliography}{10}

\bibitem{AbboudBF23}
Amir Abboud, Karl Bringmann, and Nick Fischer.
\newblock Stronger 3-sum lower bounds for approximate distance oracles via additive combinatorics.
\newblock In Barna Saha and Rocco~A. Servedio, editors, {\em 55th Annual {ACM} Symposium on Theory of Computing ({STOC} 2023)}, pages 391--404. {ACM}, 2023.
\newblock \href {https://doi.org/10.1145/3564246.3585240} {\path{doi:10.1145/3564246.3585240}}.

\bibitem{Abrahamson87}
Karl~R. Abrahamson.
\newblock Generalized string matching.
\newblock {\em {SIAM} J. Comput.}, 16(6):1039--1051, 1987.
\newblock \href {https://doi.org/10.1137/0216067} {\path{doi:10.1137/0216067}}.

\bibitem{AigerK09}
Dror Aiger and Klara Kedem.
\newblock Geometric pattern matching for point sets in the plane under similarity transformations.
\newblock {\em Inf. Process. Lett.}, 109(16):935--940, 2009.
\newblock \href {https://doi.org/10.1016/j.ipl.2009.04.021} {\path{doi:10.1016/j.ipl.2009.04.021}}.

\bibitem{Akutsu96}
Tatsuya Akutsu.
\newblock Protein structure alignment using dynamic programming and iterative improvement.
\newblock {\em IEICE TRANSACTIONS on Information and Systems}, 79(12):1629--1636, 1996.

\bibitem{AkutsuTT98}
Tatsuya Akutsu, Hisao Tamaki, and Takeshi Tokuyama.
\newblock Distribution of distances and triangles in a point set and algorithms for computing the largest common point sets.
\newblock {\em Discret. Comput. Geom.}, 20(3):307--331, 1998.
\newblock \href {https://doi.org/10.1007/PL00009388} {\path{doi:10.1007/PL00009388}}.

\bibitem{AltMWW88}
Helmut Alt, Kurt Mehlhorn, Hubert Wagener, and Emo Welzl.
\newblock Congruence, similarity, and symmetries of geometric objects.
\newblock {\em Discret. Comput. Geom.}, 3:237--256, 1988.
\newblock \href {https://doi.org/10.1007/BF02187910} {\path{doi:10.1007/BF02187910}}.

\bibitem{AmirF95}
Amihood Amir and Martin Farach.
\newblock Efficient 2-dimensional approximate matching of half-rectangular figures.
\newblock {\em Inf. Comput.}, 118(1):1--11, 1995.
\newblock \href {https://doi.org/10.1006/inco.1995.1047} {\path{doi:10.1006/inco.1995.1047}}.

\bibitem{AmirKP07}
Amihood Amir, Oren Kapah, and Ely Porat.
\newblock Deterministic length reduction: Fast convolution in sparse data and applications.
\newblock In Bin Ma and Kaizhong Zhang, editors, {\em 18th Annual Symposium on Combinatorial Pattern Matching ({CPM} 2007)}, volume 4580 of {\em Lecture Notes in Computer Science}, pages 183--194. Springer, 2007.
\newblock \href {https://doi.org/10.1007/978-3-540-73437-6\_20} {\path{doi:10.1007/978-3-540-73437-6\_20}}.

\bibitem{AmirL91}
Amihood Amir and Gad~M. Landau.
\newblock Fast parallel and serial multidimensional aproximate array matching.
\newblock {\em Theor. Comput. Sci.}, 81(1):97--115, 1991.
\newblock \href {https://doi.org/10.1016/0304-3975(91)90318-V} {\path{doi:10.1016/0304-3975(91)90318-V}}.

\bibitem{ArnoldR15}
Andrew Arnold and Daniel~S. Roche.
\newblock Output-sensitive algorithms for sumset and sparse polynomial multiplication.
\newblock In Kazuhiro Yokoyama, Steve Linton, and Daniel Robertz, editors, {\em 40th International Symposium on Symbolic and Algebraic Computation ({ISSAC 2015})}, pages 29--36. {ACM}, 2015.
\newblock \href {https://doi.org/10.1145/2755996.2756653} {\path{doi:10.1145/2755996.2756653}}.

\bibitem{BaurS83}
Walter Baur and Volker Strassen.
\newblock The complexity of partial derivatives.
\newblock {\em Theor. Comput. Sci.}, 22:317--330, 1983.
\newblock \href {https://doi.org/10.1016/0304-3975(83)90110-X} {\path{doi:10.1016/0304-3975(83)90110-X}}.

\bibitem{Boxer96}
Laurence Boxer.
\newblock Point set pattern matching in 3-d.
\newblock {\em Pattern Recognit. Lett.}, 17(12):1293--1297, 1996.
\newblock \href {https://doi.org/10.1016/0167-8655(96)00086-4} {\path{doi:10.1016/0167-8655(96)00086-4}}.

\bibitem{BringmannFN21}
Karl Bringmann, Nick Fischer, and Vasileios Nakos.
\newblock Sparse nonnegative convolution is equivalent to dense nonnegative convolution.
\newblock In Samir Khuller and Virginia~Vassilevska Williams, editors, {\em 53rd Annual {ACM} Symposium on Theory of Computing ({STOC} 2021)}, pages 1711--1724. {ACM}, 2021.
\newblock \href {https://doi.org/10.1145/3406325.3451090} {\path{doi:10.1145/3406325.3451090}}.

\bibitem{BringmannFN22}
Karl Bringmann, Nick Fischer, and Vasileios Nakos.
\newblock Deterministic and {Las Vegas} algorithms for sparse nonnegative convolution.
\newblock In Joseph~(Seffi) Naor and Niv Buchbinder, editors, {\em 33rd Annual {ACM-SIAM} Symposium on Discrete Algorithms ({SODA} 2022)}, pages 3069--3090. {SIAM}, 2022.
\newblock \href {https://doi.org/10.1137/1.9781611977073.119} {\path{doi:10.1137/1.9781611977073.119}}.

\bibitem{CardozeS98}
David~E. Cardoze and Leonard~J. Schulman.
\newblock Pattern matching for spatial point sets.
\newblock In {\em 39th Annual {IEEE} Symposium on Foundations of Computer Science ({FOCS} 1998)}, pages 156--165. {IEEE} Computer Society, 1998.
\newblock \href {https://doi.org/10.1109/SFCS.1998.743439} {\path{doi:10.1109/SFCS.1998.743439}}.

\bibitem{ChanL15}
Timothy~M. Chan and Moshe Lewenstein.
\newblock Clustered integer {3SUM} via additive combinatorics.
\newblock In Rocco~A. Servedio and Ronitt Rubinfeld, editors, {\em 47th Annual {ACM} Symposium on Theory of Computing ({STOC} 2015)}, pages 31--40. {ACM}, 2015.
\newblock \href {https://doi.org/10.1145/2746539.2746568} {\path{doi:10.1145/2746539.2746568}}.

\bibitem{Cheng05}
Qi~Cheng.
\newblock On the construction of finite field elements of large order.
\newblock {\em Finite Fields and Their Applications}, 11(3):358--366, 2005.
\newblock Ten Year Anniversary Edition!
\newblock URL: \url{https://www.sciencedirect.com/science/article/pii/S1071579705000456}, \href {https://doi.org/10.1016/j.ffa.2005.06.001} {\path{doi:10.1016/j.ffa.2005.06.001}}.

\bibitem{Cheng07}
Qi~Cheng.
\newblock Constructing finite field extensions with large order elements.
\newblock {\em {SIAM} J. Discret. Math.}, 21(3):726--730, 2007.
\newblock \href {https://doi.org/10.1137/S0895480104445514} {\path{doi:10.1137/S0895480104445514}}.

\bibitem{ChewGHKKK97}
L.~Paul Chew, Michael~T. Goodrich, Daniel~P. Huttenlocher, Klara Kedem, Jon~M. Kleinberg, and Dina Kravets.
\newblock Geometric pattern matching under euclidean motion.
\newblock {\em Comput. Geom.}, 7:113--124, 1997.
\newblock \href {https://doi.org/10.1016/0925-7721(95)00047-X} {\path{doi:10.1016/0925-7721(95)00047-X}}.

\bibitem{ChoM08}
Minkyoung Cho and David~M. Mount.
\newblock Improved approximation bounds for planar point pattern matching.
\newblock {\em Algorithmica}, 50(2):175--207, 2008.
\newblock \href {https://doi.org/10.1007/s00453-007-9059-9} {\path{doi:10.1007/s00453-007-9059-9}}.

\bibitem{CliffordC07}
Peter Clifford and Rapha{\"{e}}l Clifford.
\newblock Simple deterministic wildcard matching.
\newblock {\em Inf. Process. Lett.}, 101(2):53--54, 2007.
\newblock \href {https://doi.org/10.1016/j.ipl.2006.08.002} {\path{doi:10.1016/j.ipl.2006.08.002}}.

\bibitem{ColeH02}
Richard Cole and Ramesh Hariharan.
\newblock Verifying candidate matches in sparse and wildcard matching.
\newblock In John~H. Reif, editor, {\em 34th Annual {ACM} Symposium on Theory of Computing ({STOC} 2002)}, pages 592--601. {ACM}, 2002.
\newblock \href {https://doi.org/10.1145/509907.509992} {\path{doi:10.1145/509907.509992}}.

\bibitem{ColeHI99}
Richard Cole, Ramesh Hariharan, and Piotr Indyk.
\newblock Tree pattern matching and subset matching in deterministic {$O(n \log^3 n)$}-time.
\newblock In Robert~Endre Tarjan and Tandy~J. Warnow, editors, {\em Proceedings of the Tenth Annual {ACM-SIAM} Symposium on Discrete Algorithms, 17-19 January 1999, Baltimore, Maryland, {USA}}, pages 245--254. {ACM/SIAM}, 1999.
\newblock URL: \url{http://dl.acm.org/citation.cfm?id=314500.314565}.

\bibitem{CyganGS15}
Marek Cygan, Harold~N. Gabow, and Piotr Sankowski.
\newblock Algorithmic applications of {Baur}-{Strassen}'s theorem: {Shortest} cycles, diameter, and matchings.
\newblock {\em J. {ACM}}, 62(4):28:1--28:30, 2015.
\newblock \href {https://doi.org/10.1145/2736283} {\path{doi:10.1145/2736283}}.

\bibitem{RezendeL95}
Pedro~Jussieu de~Rezende and Der-Tsai Lee.
\newblock Point set pattern matching in d-dimensions.
\newblock {\em Algorithmica}, 13(4):387--404, 1995.
\newblock \href {https://doi.org/10.1007/BF01293487} {\path{doi:10.1007/BF01293487}}.

\bibitem{FinnKLMSVY97}
Paul~W. Finn, Lydia~E. Kavraki, Jean{-}Claude Latombe, Rajeev Motwani, Christian~R. Shelton, Suresh Venkatasubramanian, and Andrew~Chi{-}Chih Yao.
\newblock {RAPID:} randomized pharmacophore identification for drug design.
\newblock In Jean{-}Daniel Boissonnat, editor, {\em 13th Annual Symposium on Computational Geometry ({SoCG} 1997)}, pages 324--333. {ACM}, 1997.
\newblock \href {https://doi.org/10.1145/262839.262993} {\path{doi:10.1145/262839.262993}}.

\bibitem{FischerP74}
Michael~J. Fischer and Michael~S. Paterson.
\newblock String-matching and other products.
\newblock Technical report, Massachusetts Institute of Technology, USA, 1974.

\bibitem{GajentaanO95}
Anka Gajentaan and Mark~H. Overmars.
\newblock On a class of {$O(n^2)$} problems in computational geometry.
\newblock {\em Comput. Geom.}, 5:165--185, 1995.
\newblock \href {https://doi.org/10.1016/0925-7721(95)00022-2} {\path{doi:10.1016/0925-7721(95)00022-2}}.

\bibitem{GalilG88}
Zvi Galil and Raffaele Giancarlo.
\newblock Data structures and algorithms for approximate string matching.
\newblock {\em J. Complex.}, 4(1):33--72, 1988.
\newblock \href {https://doi.org/10.1016/0885-064X(88)90008-8} {\path{doi:10.1016/0885-064X(88)90008-8}}.

\bibitem{Gao99}
Shuhong Gao.
\newblock Elements of provable high orders in finite fields.
\newblock {\em Proc. Amer. Math. Soc}, 127:1615--1623, 1999.

\bibitem{GavrilovIMV99}
Martin Gavrilov, Piotr Indyk, Rajeev Motwani, and Suresh Venkatasubramanian.
\newblock Geometric pattern matching: {A} performance study.
\newblock In Victor Milenkovic, editor, {\em 15th Annual Symposium on Computational Geometry ({SoCG} 1999)}, pages 79--85. {ACM}, 1999.
\newblock \href {https://doi.org/10.1145/304893.304916} {\path{doi:10.1145/304893.304916}}.

\bibitem{Gawrychowski11}
Pawel Gawrychowski.
\newblock Pattern matching in lempel-ziv compressed strings: Fast, simple, and deterministic.
\newblock In Camil Demetrescu and Magn{\'{u}}s~M. Halld{\'{o}}rsson, editors, {\em 19th Annual European Symposium on Algorithms (ESA 2011)}, volume 6942 of {\em Lecture Notes in Computer Science}, pages 421--432. Springer, 2011.
\newblock \href {https://doi.org/10.1007/978-3-642-23719-5\_36} {\path{doi:10.1007/978-3-642-23719-5\_36}}.

\bibitem{GiorgiGC20}
Pascal Giorgi, Bruno Grenet, and Armelle~Perret du~Cray.
\newblock Essentially optimal sparse polynomial multiplication.
\newblock In Ioannis~Z. Emiris and Lihong Zhi, editors, {\em 45th International Symposium on Symbolic and Algebraic Computation ({ISSAC 2020})}, pages 202--209. {ACM}, 2020.
\newblock \href {https://doi.org/10.1145/3373207.3404026} {\path{doi:10.1145/3373207.3404026}}.

\bibitem{GoldsteinKLP16}
Isaac Goldstein, Tsvi Kopelowitz, Moshe Lewenstein, and Ely Porat.
\newblock How hard is it to find (honest) witnesses?
\newblock In Piotr Sankowski and Christos~D. Zaroliagis, editors, {\em 24th Annual European Symposium on Algorithms, {ESA} 2016, August 22-24, 2016, Aarhus, Denmark}, volume~57 of {\em LIPIcs}, pages 45:1--45:16. Schloss Dagstuhl - Leibniz-Zentrum f{\"{u}}r Informatik, 2016.
\newblock \href {https://doi.org/10.4230/LIPIcs.ESA.2016.45} {\path{doi:10.4230/LIPIcs.ESA.2016.45}}.

\bibitem{GoodrichMO94}
Michael~T. Goodrich, Joseph S.~B. Mitchell, and Mark~W. Orletsky.
\newblock Practical methods for approximate geometric pattern matching under rigid motions.
\newblock In Kurt Mehlhorn, editor, {\em 10th Annual Symposium on Computational Geometry ({SoCG} 1994)}, pages 103--112. {ACM}, 1994.
\newblock \href {https://doi.org/10.1145/177424.177572} {\path{doi:10.1145/177424.177572}}.

\bibitem{HajiaghayiSSS21}
MohammadTaghi Hajiaghayi, Hamed Saleh, Saeed Seddighin, and Xiaorui Sun.
\newblock String matching with wildcards in the massively parallel computation model.
\newblock In Kunal Agrawal and Yossi Azar, editors, {\em 33rd {ACM} Symposium on Parallelism in Algorithms and Architectures ({SPAA} 2021)}, pages 275--284. {ACM}, 2021.
\newblock \href {https://doi.org/10.1145/3409964.3461793} {\path{doi:10.1145/3409964.3461793}}.

\bibitem{HuttenlocherKK92}
Daniel~P. Huttenlocher, Klara Kedem, and Jon~M. Kleinberg.
\newblock On dynamic {Voronoi} diagrams and the minimum hausdorff distance for point sets under euclidean motion in the plane.
\newblock In David Avis, editor, {\em Proceedings of the Eighth Annual Symposium on Computational Geometry, Berlin, Germany, June 10-12, 1992}, pages 110--119. {ACM}, 1992.
\newblock \href {https://doi.org/10.1145/142675.142700} {\path{doi:10.1145/142675.142700}}.

\bibitem{Indyk97}
Piotr Indyk.
\newblock Deterministic superimposed coding with applications to pattern matching.
\newblock In {\em 38th Annual {IEEE} Symposium on Foundations of Computer Science ({FOCS} 1997)}, pages 127--136. {IEEE} Computer Society, 1997.
\newblock \href {https://doi.org/10.1109/SFCS.1997.646101} {\path{doi:10.1109/SFCS.1997.646101}}.

\bibitem{Indyk98a}
Piotr Indyk.
\newblock Faster algorithms for string matching problems: {Matching} the convolution bound.
\newblock In {\em 39th Annual {IEEE} Symposium on Foundations of Computer Science ({FOCS} 1998)}, pages 166--173. {IEEE} Computer Society, 1998.
\newblock \href {https://doi.org/10.1109/SFCS.1998.743440} {\path{doi:10.1109/SFCS.1998.743440}}.

\bibitem{IndykMV99}
Piotr Indyk, Rajeev Motwani, and Suresh Venkatasubramanian.
\newblock Geometric matching under noise: {Combinatorial} bounds and algorithms.
\newblock In Robert~Endre Tarjan and Tandy~J. Warnow, editors, {\em 10th Annual {ACM-SIAM} Symposium on Discrete Algorithms ({SODA} 1999)}, pages 457--465. {ACM/SIAM}, 1999.
\newblock URL: \url{http://dl.acm.org/citation.cfm?id=314500.314601}.

\bibitem{IndykV03}
Piotr Indyk and Suresh Venkatasubramanian.
\newblock Approximate congruence in nearly linear time.
\newblock {\em Comput. Geom.}, 24(2):115--128, 2003.
\newblock \href {https://doi.org/10.1016/S0925-7721(02)00095-0} {\path{doi:10.1016/S0925-7721(02)00095-0}}.

\bibitem{IraniR96}
Sandy Irani and Prabhakar Raghavan.
\newblock Combinatorial and experimental results for randomized point matching algorithms.
\newblock In Sue Whitesides, editor, {\em 12th Annual Symposium on Computational Geometry ({SoCG} 1996)}, pages 68--77. {ACM}, 1996.
\newblock \href {https://doi.org/10.1145/237218.237240} {\path{doi:10.1145/237218.237240}}.

\bibitem{JinX23}
Ce~Jin and Yinzhan Xu.
\newblock Removing additive structure in 3sum-based reductions.
\newblock In Barna Saha and Rocco~A. Servedio, editors, {\em 55th Annual {ACM} Symposium on Theory of Computing ({STOC} 2023)}, pages 405--418. {ACM}, 2023.
\newblock \href {https://doi.org/10.1145/3564246.3585157} {\path{doi:10.1145/3564246.3585157}}.

\bibitem{Kalai02}
Adam Kalai.
\newblock Efficient pattern-matching with don't cares.
\newblock In David Eppstein, editor, {\em 13th Annual {ACM-SIAM} Symposium on Discrete Algorithms ({SODA} 2002)}, pages 655--656. {ACM/SIAM}, 2002.
\newblock URL: \url{http://dl.acm.org/citation.cfm?id=545381.545468}.

\bibitem{KaltofenL88}
Erich Kaltofen and Yagati~N. Lakshman.
\newblock Improved sparse multivariate polynomial interpolation algorithms.
\newblock In Patrizia~M. Gianni, editor, {\em 13th International Symposium on Symbolic and Algebraic Computation ({ISSAC 1988})}, volume 358 of {\em Lecture Notes in Computer Science}, pages 467--474. Springer, 1988.
\newblock \href {https://doi.org/10.1007/3-540-51084-2\_44} {\path{doi:10.1007/3-540-51084-2\_44}}.

\bibitem{Kosaraju89}
S.~Rao Kosaraju.
\newblock Efficient tree pattern matching.
\newblock In {\em 30th Annual {IEEE} Symposium on Foundations of Computer Science ({FOCS} 1989)}, pages 178--183. {IEEE} Computer Society, 1989.
\newblock \href {https://doi.org/10.1109/SFCS.1989.63475} {\path{doi:10.1109/SFCS.1989.63475}}.

\bibitem{LagariasO87}
Jeffrey~C. Lagarias and Andrew~M. Odlyzko.
\newblock Computing pi(x): An analytic method.
\newblock {\em J. Algorithms}, 8(2):173--191, 1987.
\newblock \href {https://doi.org/10.1016/0196-6774(87)90037-X} {\path{doi:10.1016/0196-6774(87)90037-X}}.

\bibitem{Li00}
Lei Li.
\newblock On the arithmetic operational complexity for solving {Vandermonde} linear equations.
\newblock {\em Japan Journal of Industrial and Applied Mathematics}, 17(15), 2000.
\newblock \href {https://doi.org/10.1007/BF03167332} {\path{doi:10.1007/BF03167332}}.

\bibitem{Morgenstern85}
Jacques Morgenstern.
\newblock How to compute fast a function and all its derivatives: {A} variation on the theorem of {Baur}-{Strassen}.
\newblock {\em {SIGACT} News}, 16(4):60--62, 1985.
\newblock \href {https://doi.org/10.1145/382242.382836} {\path{doi:10.1145/382242.382836}}.

\bibitem{MountNM98}
David~M. Mount, Nathan~S. Netanyahu, and Jacqueline~Le Moigne.
\newblock Improved algorithms for robust point pattern matching and applications to image registration.
\newblock In Ravi Janardan, editor, {\em 14th Annual Symposium on Computational Geometry ({SoCG} 1998)}, pages 155--164. {ACM}, 1998.
\newblock \href {https://doi.org/10.1145/276884.276902} {\path{doi:10.1145/276884.276902}}.

\bibitem{Muthukrishnan95}
S.~Muthukrishnan.
\newblock New results and open problems related to non-standard stringology.
\newblock In Zvi Galil and Esko Ukkonen, editors, {\em 6th Annual Symposium on Combinatorial Pattern Matching ({CPM} 1995)}, volume 937 of {\em Lecture Notes in Computer Science}, pages 298--317. Springer, 1995.
\newblock \href {https://doi.org/10.1007/3-540-60044-2\_50} {\path{doi:10.1007/3-540-60044-2\_50}}.

\bibitem{MuthukrishnanP94}
S.~Muthukrishnan and Krishna~V. Palem.
\newblock Non-standard stringology: algorithms and complexity.
\newblock In Frank~Thomson Leighton and Michael~T. Goodrich, editors, {\em Proceedings of the Twenty-Sixth Annual {ACM} Symposium on Theory of Computing, 23-25 May 1994, Montr{\'{e}}al, Qu{\'{e}}bec, Canada}, pages 770--779. {ACM}, 1994.
\newblock \href {https://doi.org/10.1145/195058.195457} {\path{doi:10.1145/195058.195457}}.

\bibitem{MuthukrishnanR95}
S.~Muthukrishnan and H.~Ramesh.
\newblock String matching under a general matching relation.
\newblock {\em Inf. Comput.}, 122(1):140--148, 1995.
\newblock \href {https://doi.org/10.1006/inco.1995.1144} {\path{doi:10.1006/inco.1995.1144}}.

\bibitem{Nakos20}
Vasileios Nakos.
\newblock Nearly optimal sparse polynomial multiplication.
\newblock {\em {IEEE} Trans. Inf. Theory}, 66(11):7231--7236, 2020.
\newblock \href {https://doi.org/10.1109/TIT.2020.2989385} {\path{doi:10.1109/TIT.2020.2989385}}.

\bibitem{Pan01}
Victor~Y. Pan.
\newblock {\em Structured Matrices and Polynomials: {Unified} Superfast Algorithms}.
\newblock Springer-Verlag, Berlin, Heidelberg, 2001.

\bibitem{Pinter85}
Ron~Y. Pinter.
\newblock Efficient string matching with don't-care patterns.
\newblock In {\em Combinatorial Algorithms on Words}, pages 11--29. Springer Berlin Heidelberg, 1985.

\bibitem{Rucklidge93}
William Rucklidge.
\newblock Lower bounds for the complexity of the hausdorff distance.
\newblock In {\em 5th Canadian Conference on Computational Geometry ({CCCG} 1993)}, pages 145--150. University of Waterloo, 1993.

\bibitem{Rucklidge96}
William Rucklidge.
\newblock {\em Efficient Visual Recognition Using the Hausdorff Distance}, volume 1173 of {\em Lecture Notes in Computer Science}.
\newblock Springer, 1996.
\newblock \href {https://doi.org/10.1007/BFb0015091} {\path{doi:10.1007/BFb0015091}}.

\bibitem{Rumelhart86}
David~E. Rumelhart, Geoffrey~E. Hinton, and Ronald~J. Williams.
\newblock Learning representations by back-propagating errors.
\newblock {\em Nature}, 323(6088):533--536, 1986.
\newblock \href {https://doi.org/10.1038/323533a0} {\path{doi:10.1038/323533a0}}.

\bibitem{Shoup90}
Victor Shoup.
\newblock Searching for primitive roots in finite fields.
\newblock In Harriet Ortiz, editor, {\em 22nd Annual {ACM} Symposium on Theory of Computing ({STOC} 1990)}, pages 546--554. {ACM}, 1990.
\newblock \href {https://doi.org/10.1145/100216.100293} {\path{doi:10.1145/100216.100293}}.

\bibitem{Shparlinski96}
Igor~E. Shparlinski.
\newblock On finding primitive roots in finite fields.
\newblock {\em Theor. Comput. Sci.}, 157(2):273--275, 1996.
\newblock \href {https://doi.org/10.1016/0304-3975(95)00164-6} {\path{doi:10.1016/0304-3975(95)00164-6}}.

\bibitem{TaoCH12}
Terence Tao, Ernest~Croot III, and Harald Helfgott.
\newblock Deterministic methods to find primes.
\newblock {\em Math. Comput.}, 81(278):1233--1246, 2012.
\newblock \href {https://doi.org/10.1090/S0025-5718-2011-02542-1} {\path{doi:10.1090/S0025-5718-2011-02542-1}}.

\bibitem{Ukkonen10}
Esko Ukkonen.
\newblock Geometric point pattern matching in the {Knuth-Morris-Pratt} way.
\newblock {\em J. Univers. Comput. Sci.}, 16(14):1902--1911, 2010.
\newblock \href {https://doi.org/10.3217/jucs-016-14-1902} {\path{doi:10.3217/jucs-016-14-1902}}.

\bibitem{vonzurGathenG13}
Joachim von~zur Gathen and Jürgen Gerhard.
\newblock {\em Modern Computer Algebra}.
\newblock Cambridge University Press, 3rd edition, 2013.
\newblock \href {https://doi.org/10.1017/CBO9781139856065} {\path{doi:10.1017/CBO9781139856065}}.

\bibitem{Werbos74}
Paul~J. Werbos.
\newblock {\em Beyond regression: {New} Tools for prediction and analysis in the behavioral science}.
\newblock PhD thesis, Harvard University, 1974.

\bibitem{Werbos94}
Paul~J. Werbos.
\newblock {\em The Roots of Backpropagation: {From} Ordered Derivatives to Neural Networks and Political Forecasting}.
\newblock Wiley-Interscience, USA, 1994.

\end{thebibliography}

\end{document}